\documentclass[epj]{svjour}
\usepackage{graphicx}
\usepackage{amsmath}
\usepackage{amsfonts}
\usepackage{amssymb}

 \usepackage{balance}
 \usepackage{hyperref}
 \usepackage{multirow}
\usepackage{caption}
\usepackage{subcaption}
\begin{document}
\title{The genesis of two-hump, W-shaped and M-shaped soliton propagations of the coupled  Schr$\ddot{\text{o}}$dinger-Boussinesq equations with conformable derivative}
\author{Prakash Kumar Das %\inst{1} \and Second author\inst{2}% etc
\thanks{\emph{e-mail:} prakashdas.das1@gmail.com}%
}                     % Do not remove
\offprints{}          % Insert a name or remove this line
\institute{Department of Mathematics, Trivenidevi Bhalotia College, Raniganj, Paschim Bardhaman,
West Bengal, India-713347 %\and the second here
}
\date{Received: date / Revised version: date}
\abstract{
This work oversees with the coupled  Schr$\ddot{\text{o}}$dinger-Boussinesq equations with conformable derivative, which have lots of applications in laser and plasma.  The said equations are  reduced to a coupled stationary form using complex travelling wave transformation. Next  Painlev$\acute{\text{e}}$ test applied to derived the integrable cases of the reduced equation, after that using RCAM  derived the solution of reduced equations integrable and nonintegrable cases. Few theorems have been presented and proved to ensure their boundedness. All presented boundedness cases have been checked and explained by plotting the solutions for particulars values of parameters satisfying them.  The obtained solutions of stationary form utilized to derive solutions of the coupled  Schr$\ddot{\text{o}}$dinger-Boussinesq equations with conformable
derivative. The  derived solutions have been plotted and explained. From this, it appears that these solutions propagate by maintaining their two-hump, W-shaped, M-shaped solutions shapes.
%
%\PACS{
     % {47.35.Fg}{Solitary waves}   \and {47.27.De}{Coherent structures}   \and {02.70.−c}{Computational techniques; simulations} \and {02.60.Cb}{ Numerical simulation; solution of equations} \and {02.30.Jr}{Partial differential equations} \and 
   %   {94.20.Bb}{ Wave propagation}
   %  } % end of PACS codes
} %end of abstract
\titlerunning{Two-hump soliton propagations of the coupled  SB equations with conformable derivative}
\authorrunning{P.K.Das}
\maketitle
\section{Introduction}
Diverse wave propagations observed in abundant fields of the environment are  modelled by non-linear ordinary or partial differential equations (ODEs/PDEs) involving integer or fractional order derivative.  Their chaotic features can be explained by the solutions of the equations which describes them. Due to the existence  of interaction terms in such equations, it is not always an easy task to derive an exact solution to these equations. Despite that, a vast amount of literature exists for constructing exact travelling wave and soliton solutions \cite{li2020generalized,yu2019inverse,yu2020nonstandard,olver2000applications,park2020dynamical,gao2020novel,rezazadeh2019large,savaissou2020exact} and  approximation solutions \cite{malfliet1992solitary,abbasbandy2009homotopy,xinhui2012homotopy,adomian1994solving,duan2011new,adomian1983inversion,adomian1993analytic,adomian1993new,adomian1994modified,wazwaz2000approximate,wazwaz2000modified,wazwaz2001reliable,wazwaz2001numerical,wazwaz2002numerical} of ODEs and PDEs involving integer derivative. Also there exists few direct methods \cite{khalil2014new,pandir2018analytical,ellahi2018exact,sabi2019new,houwe2020solitary} to derive exact solutions and and analytical methods  \cite{yong2008numerical,atangana2013time,biswas2019approximate,osman2019new,arqub2020numerical,kumar2019hybrid,goswami2019efficient,prakash2019numerical}  to deal with approximate solutions of   nonlinear fractional differential equations. These direct methods often need to guess the syntax of solutions and solve a system of nonlinear algebraic equations. Thus these schemes fail when the choice of a form of solution is not suitable or unable to solve the system of equations. And the above mentioned approximate methods are often observed to have slow convergent rate and always unable to provide the close form of the series solution. To overcome the above-mentioned limitations of those methods we adopt the rapidly convergent approximation method (RCAM) \cite{das2020new,das2020chirped,das2019rapidly,das2019some,das2018piecewise,das2018solutions,das2018rapidly,das2016rapidly,das2015improved}. This current work  deals with this scheme to obtained some new  solutions for the coupled  Schr$\ddot{\text{o}}$dinger-Boussinesq systems (SBS) with conformable derivative.

To attained the goal, the paper is organised as follows: in sections \ref{sec2a} and \ref{sec2} we present the basic properties of conformal derivative and methodology of RCAM to solve a system of ODEs respectively. In section  \ref{sec3} we reduced SBS to coupled ODEs by employing a complex wave transformation. The Painlev$\acute{\text{e}}$ test has been employed to the reduced coupled ODEs and identified its integrable cases in section \ref{sec4}. In section \ref{sec5}, the integrable and nonintegrable cases of the reduced ODEs have been solved by RCAM. Few theorems have been presented to study their boundedness and utilised to plot the stationary form of the solutions. These solutions used to derive explicit solutions of SBS. Furthermore, the main characteristics of these derived solutions are graphically discussed.  We summarize our outlook on the present work in section \ref{sec6}.  		
\section{Properties of conformable derivative}\label{sec2a}
This portion deals with few definitions and properties of conformable derivative \cite{khalil2014new,abdeljawad2015conformable}:
%\textbf{Definition:}
\begin{definition}The conformable  derivative of a function $f : [0,\infty ) \rightarrow \mathbb{R}$ of order $\alpha$ is defined by
\begin{eqnarray}\label{FDe1}
T_{\alpha}(h)(t)= \lim_{\epsilon \rightarrow 0} \frac{h\left( t+\epsilon\ t^{1-\alpha} \right)-h(t)}{\epsilon},
\end{eqnarray}
for all $t>0, \ \alpha \in (0,1).$ If f is $\alpha$-differentiable in some $(0, a), a > 0,$ and $ \lim_{t \rightarrow 0+} f^{(\alpha)}(t)$ exists, then define
 $f^{(\alpha)}(0)=\lim_{t \rightarrow 0+} f^{(\alpha)}(t).$
\end{definition}
 In the corresponding sections of this paper we replace  the notation  $T_{\alpha}(f)(t)$ by  $f^{(\alpha)}(t)$, to prevail the conformable  derivatives of $f$ of order $\alpha$. Some notable features of conformable derivative are listed  below:\\
If  $\alpha \in (0,1],$ and $f, g$ be $\alpha$-differentiable at a point $t>0$ then we have
\begin{itemize}
\item[1.]   $T_{\alpha} \left( a \ f+b\ g\right)=a\ T_{\alpha}(f)+b\ T_{\alpha}(g),$ for all $a,b \in \mathbb{R}.$ \\
\item[2.]  $T_{\alpha}\left( t^p \ \right)=p t^{p-\alpha}$ for all $p \in \mathbb{R}.$\\
\item[3.]  $T_{\alpha}(\lambda)=0$, for all constant functions $f(t)=\lambda.$\\
\item[4.]  $T_{\alpha}\left(f g \right)=f\ T_{\alpha}(g)+g\ T_{\alpha}(f).$ \\
\item[5.]  $T_{\alpha}\left(\frac{f}{g}\right)=\frac{g\ T_{\alpha}(f)+f\ T_{\alpha}(g)}{g^2}$.\\
\item[6.] if $f$ is differentiable, then  $T_{\alpha}(f)(t)=t^{1-\alpha} \frac{df}{dt}(t).$
\end{itemize}
\begin{theorem}[Chain Rule \cite{abdeljawad2015conformable,eslami2016first,eslami2016exact,chen2018simplest}] Let $f,g : (0 , \infty)\rightarrow \mathbb{R}$ be two  differentiable  functions and also $f$ is $\alpha$-differentiable, then, one has the following rule:\\
 $T_{\alpha}( fog ) ( t ) = t^{ 1-\alpha} g'( t ) f'( g ( t ) ).$
 \end{theorem}

\section{Basic methodology of RCAM} \label{sec2}
Take into account a system of ODEs 
\begin {eqnarray}\label{eq2p1}
 {\cal X}^{''}-{\cal A}^{2}.\; {\cal X}={\cal N},
\end{eqnarray}
where   ${\cal X }, {\cal A }$ and ${\cal N}$  are the matrix of   dependent variables, constant coefficients and interaction terms  respectively, having form 
 $${\cal X}=\left[\begin{array}{c}U_1(x)\\U_2(x) \\ \vdots \\U_k(x) \\\end{array}\right] , \ \ {\cal A}=\left[\begin{array}{c}
 \lambda_1\ \ 0\ \ \cdots 0 \\ 0\ \ \lambda_2 \ \ \cdots 0 \\ \vdots \\ 0\ \ 0\ \ \cdots \ \ \lambda_k  \\\end{array}\right], \ \text{and}  \ {\cal N}=\left[\begin{array}{c}
 N_1\left(U_1(x),\cdots,U_k(x)\right) \\ N_2\left(U_1(x),\cdots,U_k(x)\right) \\ \vdots \\N_k\left(U_1(x),\cdots,U_k(x)\right) \\\end{array}\right].$$ 
To solve  (\ref{eq2p1}), we remodel it  in an exponential matrix operator conformation
\begin {equation}\label{eq2p2}
{\cal O}[{\cal X}](x) = {\cal N},
\end{equation}
where linear exponential matrix operator can be recast in the form
\begin {equation}\label{eq2p3}
\hat{{\cal O}}[\cdot](x)  = e^{ {\cal A}.x}\frac{d}{dx}\left(e^{- 2 {\cal A}.x}\frac{d}{dx}\left(e^{ {\cal A}.x}[\cdot]\right)\right).
\end{equation}
The inverse operator $\hat{{\cal O}^{-1}}$ of  ${\cal O}[](x)$ is  conferred by
\begin{eqnarray}\label{eq2p4}
\hat{{\cal O}}^{-1}[\cdot](x)  =  e^{- {\cal A}.x}\int^{x} e^{2 {\cal A}.x^{'}}\int^{x^{'}} e^{- {\cal A}.x^{''}}[\cdot] dx^{''}  dx^{'}.
\end{eqnarray}
Operating ${\cal O}^{-1}$ on ${\cal O}[ {\cal X}](x)$ and employing integration by parts yields
\begin{eqnarray}\label{eq2p5}
\hat{{\cal O}}^{-1}\left( {\cal X}^{''}- {\cal A}^{2}.\;  {\cal X} \right)  = {\cal X}- {\cal C}. e^{ {\cal A}. x}- {\cal D}.e^{- {\cal A}. x},
\end{eqnarray}
where ${\cal C}=\left[\begin{array}{c}c_1 \\c_2 \\ \vdots \\c_k \\\end{array}\right]$ and ${\cal D}=\left[\begin{array}{c}
 d_1 \\  d_2 \\ \vdots  \\d_k \\\end{array}\right]$  are matrices of integration constants.
Operating ${\cal O}^{-1}$  on (\ref{eq2p2}) and utilising (\ref{eq2p5}), provides
\begin{equation}\label{eq2p6}
{\cal X} ={\cal C}.e^{{\cal A}. x} +{\cal D}.e^{-{\cal A}. x}+\hat{\cal O}^{-1}[{\cal N}](x),
\end{equation}
where ${\cal C}$, and  ${\cal D}$ are three arbitrary constants matrices. To derive the unknown ${\cal X}$ terms in the R.H.S of (\ref{eq2p6}), we recast them in the syntax
\begin{eqnarray}\label{eq2p7}
{\cal X} & \cong &\left[\begin{array}{c}U_1(x) \\ U_2(x) \\ \vdots  \\U_k(x) \\\end{array}\right]= \sum_{m=0}^\infty \left[\begin{array}{c}U_{1,m}(x) \\ U_{2,m}(x) \\ \vdots  \\U_{k,m}(x) \\\end{array}\right]
\end{eqnarray}
and
${\cal N}=\sum_{m=0}^\infty \Delta_m (x),$
with
\begin {eqnarray}\label{eq2p8}
 \Delta_m \cong \left[\begin{array}{c}\Delta_{1,m}(x) \\ \Delta_{2,m}(x) \\ \vdots  \\\Delta_{k,m}(x) \\ \end{array}\right] = \frac{1}{m!}\left[\frac{d^m}{d\epsilon^m}\left[\begin{array}{c}N_{1}(x) \\N_{2}(x) \\ \vdots \\N_{k}(x) \\\end{array}\right]\right]_{\epsilon=0}
\end{eqnarray}
and $N_{j}(x), \ j=1,2, \cdots,k$ are given by  $$N_{j}(x)=\left(\sum_{k=0}^\infty U_{1,k}\epsilon^k,\sum_{k=0}^\infty U_{2,k}\epsilon^k, \cdots  ,\sum_{k=0}^\infty U_{k,k}\epsilon^k \right).$$ The terms $\Delta_{i,m} (x)= \Delta_{i,m} (U_{1,0}(x),U_{1,1}(x), …..,U_{1,m}(x), \cdots, U_{k,0}(x),U_{k,1}(x),$ $…..,U_{k,m}(x)),$ $ i=1,2,\cdots ,k$ are  Adomian polynomials \cite{adomian1994solving,duan2011new,adomian1983inversion,adomian1993analytic,adomian1993new,adomian1994modified} outturn from the formula (\ref{eq2p8}).
Use of (\ref{eq2p8}) in (\ref{eq2p9}) provides
\begin {eqnarray}\label{eq2p11}
{\cal X} =  \left[ \begin{array}{c} c_{1}\; e^{\lambda_1 x}+ d_{1} \; e^{-\lambda_1 x} \\ c_{2}\; e^{\lambda_2 x}+ d_{2} \; e^{-\lambda_2 x} \\ \vdots \\ c_{k}\; e^{\lambda_k x}+ d_{k} \; e^{-\lambda_k x}\\\end{array}\right] + \hat{\cal O}^{-1}\left[\sum_{m=0}^\infty  \left[\begin{array}{c}\Delta_{1,m}(x) \\ \Delta_{2,m}(x) \\ \vdots  \\\Delta_{k,m}(x) \\ \end{array}\right] \right]. \nonumber \\
\end{eqnarray}
We obey the footsteps of \cite{das2019rapidly}, to get the higher order iteration terms as
\begin {eqnarray}\label{eq2p9}
&{\cal X}_0 \cong \left[\begin{array}{c}U_{1,0}(x)\\U_{2,0}(x)\\ \vdots \\U_{k,0}(x) \\\end{array}\right]= \left[\begin{array}{c} c_{1}\; e^{\lambda_1 x}+ d_{1} \; e^{-\lambda_1 x}\\ c_{2}\; e^{\lambda_2 x}+ d_{2} \; e^{-\lambda_2 x} \\ \vdots \\ c_{k}\; e^{\lambda_k x}+ d_{k} \; e^{-\lambda_k x} \\\end{array}\right],
\end{eqnarray}
\begin {eqnarray}\label{eq2p10}
&  {\cal X}_{n+1} \cong \left[\begin{array}{c}U_{1,n+1}(x)\\ U_{2,n+1}(x)\\  \vdots \\U_{k,n+1}(x) \\\end{array}\right]= \hat{\cal O}^{-1}\left[\begin{array}{c}\Delta_{1,n}(x) \\ \Delta_{2,n}(x) \\ \vdots  \\\Delta_{k,n}(x) \\ \end{array}\right],
\end{eqnarray}
$n \geq 0$. In case $ \lambda_{i} > 0 $, treatment of the vanishing boundary condition $U_{i}(\infty) = 0 $  in (\ref{eq2p11}) for  the localized solution leads us to $c_{i} = 0,\ i=1,2, \cdots, k $. Therefore the leading and higher order iteration terms  of the series solution are produced by (\ref{eq2p10}) and
\begin {eqnarray}\label{eq2p14}
&{\cal X}_0 \cong \left[\begin{array}{c}U_{1,0}(x)\\U_{2,0}(x)\\ \vdots \\U_{k,0}(x) \\\end{array}\right]= \left[\begin{array}{c}  d_{1} \; e^{-\lambda_1 x} \\  d_{2} \; e^{-\lambda_2 x} \\ \vdots \\  d_{k} \; e^{-\lambda_k x} \\\end{array}\right].
\end{eqnarray}
Further to obtain the localized solution in case $ \lambda_{i} < 0 $, for the boundary condition $U_{i}(-\infty) = 0, $   we go ahead with restraining the term involving $e^{\lambda_i x}$. In this case, the commanding and subsequent  terms  of the solution are yields by (\ref{eq2p10}) with
\begin {eqnarray}\label{eq2p15}
&{\cal X}_0 \cong \left[\begin{array}{c}U_{1,0}(x)\\U_{2,0}(x)\\ \vdots \\U_{k,0}(x) \\\end{array}\right]= \left[\begin{array}{c} c_{1}\; e^{\lambda_1 x} \\ c_{2}\; e^{\lambda_2 x}\\ \vdots \\ c_{k}\; e^{\lambda_k x} \\\end{array}\right].
\end{eqnarray}
One can easily obtain the iterative terms and the general term of the series  (or generating function) by taking advantage of symbolic software. That leads to the exact solution of the discussed system of ODEs.  
\section{Solution of SBS equations}\label{sec3}
Consider the generalized Schr$\ddot{o}$dinger-Boussinesq system (SBS) \cite{liao2020two,baleanu2019investigation,chowdhury1998painleve} in the form
\begin{eqnarray}\label{SBs0}
&& i \left( \frac{\partial E}{\partial t}+\delta_1 \  \frac{\partial E}{\partial x}\right)+\delta_2 \  \frac{\partial^2 E}{\partial x^2}= \delta_3 \ N\ E,  \nonumber \\
&& \frac{\partial^2 N}{\partial t^2}+\mu_1\ \frac{\partial^2 N}{\partial x^2}+\mu_2\ \frac{\partial^4 N}{\partial x^4}+\mu_3\ \frac{\partial^2 N^2}{\partial x^2}=\mu_4\ \frac{\partial^2 |E|^2}{\partial x^2},
\end{eqnarray}
where $E(x,t)$ is complex wave field, $N(x,t)$ is real wave field,  $\delta_i, \ i=1,2,3$ and $\mu_j, \ j=1,2,3,4$ are arbitrary (real) parameters. The SBS studied in stationary propagation of coupled nonlinear magnetosonic waves and upper-hybrid in amagnetized plasmas \cite{rao1989exact,baleanu2019investigation}, in the field of laser and plasma it describes the interaction of long waves with short wave packets in nonlinear dispersive media \cite{makhankov1974stationary}, it plays important role in diatomic lattice system \cite{yajima1979soliton}, and  Langmuir soliton formation \cite{rao1996coupled,rao1997coupled}. Several researchers applied Lie point symmetry \cite{rao1996coupled}, Painlev$\acute{\text{e}}$ Analysis and Backlund transformations \cite{chowdhury1998painleve} finite difference schemes \cite{liao2020two},  simplest equation method \cite{neirameh2015topological}, extended trial equation method \cite{gepreel2016extended}, and  direct algebraic method \cite{eslami2015soliton,hon2009series,fan2003algebraic} to study its  properties and solutions. 

In this work, we consider the  coupled  Schrödinger-Boussinesq  systems (SBS) with conformable derivative
\begin{eqnarray}\label{SBs}
&& i \left( E^{(\beta)}_t+\delta_1 \  E^{(\alpha)}_x \right)+\delta_2 \  E^{(\alpha)}_{xx}= \delta_3 \ N\ E,  \nonumber \\
&& N^{(\beta)}_{tt}+\mu_1\ N^{(\alpha)}_{xx}+\mu_2\ N^{(\alpha)}_{xxxx}+ \mu_3\ \left( N^2 \right)^{(\alpha)}_{xx}= \mu_4\ \left( |E|^2 \right)^{(\alpha)}_{xx},
\end{eqnarray}
where $f^{(\alpha)}_{*}$ and $f^{(\beta)}_{*}$ represents $\alpha$ and $\beta$ order conformal  derivative of $f$ with respect to suffix variables respectively. 
Here our aim is to study its integrability and derive its new stationary different shaped exact solutions by employing RCAM. 
To attain the target we impose the travelling wave transformation  \cite{inc2018dark,eslami2016exact,chen2018simplest}
\begin{eqnarray}\label{SBsTr}
&& E(x, t)=u(\xi)\ e^{i \eta},  \ \ \eta=k_1\ \frac{x^{\alpha}}{\alpha}+k_2\ \frac{t^{\beta}}{\beta}+c_0, \nonumber \\ 
&& N(x,t)=v(\xi),\ \ \xi=k_3 \ \frac{x^{\alpha}}{\alpha}+c\ \frac{t^{\beta}}{\beta}+\xi_0,
\end{eqnarray}
 to (\ref{SBs}) and equating real and imaginary parts we get
\begin{eqnarray}\label{reimpp}
&& c+k_3\ \delta_1+2 k_1 \ k_3 \ \delta_2=0, \nonumber \\
&& k_3^2 \ \delta_2\ u''-(k_2+k_1 \delta_1+k_1^2\ \delta_2) u-\delta_3\ u\ v=0, \\
&& k_3^4 \  \mu_2\ v^{(4)}+(c^2+\mu_1 \ k_3^2) v''-2 k_3^2 \ \mu_4 (uu')'+2 k_3^2 \ \mu_3\ (v v')'=0. \nonumber 
\end{eqnarray}
Integrating the last equation of the system twice and taking integration constant zero,  reduces system (\ref{reimpp}) to 

\begin{eqnarray}\label{meq1}
&&u''(\xi )-\lambda_1^2 \ u(\xi )=\alpha_1    \ u(\xi )\ v(\xi ), \nonumber \\
&&v''(\xi )-\lambda_2^2 \ v(\xi )=\beta_1    \ u(\xi )^2+  \gamma_1    \  v(\xi )^2,
\end{eqnarray}
 for the values of the parameters
\begin{eqnarray}\label{reimp}
&& c=- k_3 ( \delta_1+2 k_1 \ \delta_2) ,\ \alpha_1    =\frac{\delta_3}{\delta_2 \ k_3^2},\ \beta_1   =\frac{\mu_4}{\mu_2 \ k_3^2} ,\ \gamma_1    =-\frac{\mu_3}{\mu_2 \ k_3^2}, \nonumber \\
&& \lambda_1^2=\frac{k_1 \left(\delta_1+\delta_2 k_1\right)+k_2}{\delta _2 \ k_3^2}, \ \  \lambda_2^2=-\frac{\delta_1^2+4 \delta_1 \delta_2  k_1+4  k_1^2 \delta_2^2+\mu_1}{\mu_2 \ k_3^2}. 
\end{eqnarray}
Here $u(\xi)$ and $v(\xi)$ are real fields, $\xi$ is the (real) independent variable and all the other
remaining quantities are free parameters. It is important to note here that the equations are invariant under the transformations (i) $\xi \rightarrow -\xi$
 (ii)$u \rightarrow -u$, and (iii) $\xi \rightarrow \xi+C$, where $C$ is a constant.
 
\section{ The Painlev$\acute{\text{e}}$ test of Eq.(\ref{meq1})}\label{sec4}
The existence of solutions is a necessary condition of integrability, but it is not sufficient.  To confirm the integrability other tests such
as the Lax pair or the Painlev$\acute{\text{e}}$ test should be used for the proposed model. The Painlev$\acute{\text{e}}$ analysis is a powerful scheme to check the integrability of a system. Here we apply this  important tools Mathematica package PainleveTest.m \cite{baldwin2006symbolic} to examine the integrability of equation (\ref{meq1}).
Application of the package yields  the resonances  of the considered equation as
 $$-1;\ 6;\  -\frac{\sqrt{-\alpha_1     (23 \alpha_1    -48 \gamma_1     )}-5 \alpha_1    }{2 \alpha_1    };\ \frac{\sqrt{-\alpha_1     (23 \alpha_1    -48 \gamma_1     )}+5 \alpha_1    }{2 \alpha_1    }.$$
 It can be checked that for  resonances $-1; 6$,
this model (\ref{meq1}) fails the Painlev$\acute{\text{e}}$ test, and the remaining two resonances are symbolic so we can not proceed further. To move forward we assume that these symbolic resonances are equal to positive integer  $k$ (say), which yields
 $$ \pm \frac{\sqrt{-\alpha_1     (23 \alpha_1    -48 \gamma_1     )}\pm 5 \alpha_1    }{2 \alpha_1    }=k \ \ \text{or}\ \ \gamma_1    =\frac{1}{12} \left(12  +  k^2-5  k\right) \alpha_1   .$$
Subsequent, setting different positive integer values for $k$ and  using the same Mathematica package, we get the following two integrable cases;\\
1. For $k=8$ compatibility condition is $\gamma_1    =3 \alpha_1   .$\\
2. For $k=2 \ \text{or}\ 3$ compatibility condition is $\gamma_1    =\frac{\alpha_1   }{2} , \ \ \ \lambda_1=\lambda_2.$
\section{Solution of SBS to some special integrable and nontegrable cases by RCAM}\label{sec5}
 In this section we solve  above presented  integrable and one nonintegrable cases  of (\ref{meq1}) (or (\ref{SBs})) by RCAM.
\subsection{Case-I \  $\gamma_1    =3 \alpha_1   $}
 In this case, applying RCAM we get following correction terms
 \begin{eqnarray*}
% \begin{cases}
&& \begin{cases}
  u_0(\xi )=u_- \ e^{\lambda _1 \xi },\\
  v_0(\xi )=v_- \ e^{\lambda _2 \xi }
 \end{cases} \\
&&  \begin{cases}
  u_1(\xi )=\frac{\alpha_1     u_- v_- e^{\left(\lambda _1+\lambda _2\right) \xi }}{\lambda _2 \left(2 \lambda _1+\lambda
   _2\right)},\\
  v_1(\xi )=\frac{\beta_1     u_-^2 e^{2 \lambda _1 \xi }}{4 \lambda _1^2-\lambda _2^2}+\frac{\alpha_1     v_-^2 e^{2 \lambda _2 \xi
   }}{\lambda _2^2}
 \end{cases} \\
 && \begin{cases}
  u_2(\xi )=\frac{\alpha_1     u_- e^{\lambda _1 \xi } \left(\beta_1     \lambda _2^3 u_-^2 e^{2 \lambda _1 \xi }-4 \alpha_1     \lambda
   _1^2 \left(\lambda _2-2 \lambda _1\right) v_-^2 e^{2 \lambda _2 \xi }\right)}{32 \lambda _1^4 \lambda
   _2^3-8 \lambda _1^2 \lambda _2^5},\\
  v_2(\xi )=\frac{\alpha_1     v_- e^{\lambda _2 \xi } \left(3 \alpha_1     \lambda _1 \left(\lambda _2^2-4 \lambda _1^2\right)
   v_-^2 e^{2 \lambda _2 \xi }-4 \beta_1     \lambda _2^3 u_-^2 e^{2 \lambda _1 \xi }\right)}{4 \lambda _1 \lambda
   _2^4 \left(\lambda _2^2-4 \lambda _1^2\right)},
 \end{cases}
    \end{eqnarray*}
  \begin{align*}
  \begin{cases}
  u_3(\xi )=&\frac{\alpha_1    ^2 u_- v_- e^{\left(\lambda _1+\lambda _2\right) \xi } \left(\beta_1     \lambda _2^4 \left(6 \lambda
   _1+\lambda _2\right) u_-^2 e^{2 \lambda _1 \xi }-2 \alpha_1     \lambda _1^2 \left(\lambda _2-2 \lambda
   _1\right) \left(2 \lambda _1+\lambda _2\right)^2 v_-^2 e^{2 \lambda _2 \xi }\right)}{8 \lambda _1^2
   \left(2 \lambda _1-\lambda _2\right) \lambda _2^5 \left(2 \lambda _1+\lambda _2\right)^3},\\
  v_3(\xi )=&\alpha_1     \left[\beta_1    ^2 \lambda _2^6 u_-^4 \left(2 \lambda _1+\lambda _2\right) e^{4 \lambda _1 \xi }-8 \alpha_1     \beta_1     \lambda _1
   \left(\lambda _2-2 \lambda _1\right) \left(\lambda _1+\lambda _2\right)^2 \lambda _2^3 u_-^2 \right. \\
 &  \left. \times v_-^2 e^{2
   \left(\lambda _1+\lambda _2\right) \xi }+2 \alpha_1^2 \lambda _1^2 \left(2 \lambda _1+\lambda _2\right)\left(\lambda
   _2^2-4 \lambda _1^2\right)^2 v_-^4 e^{4 \lambda _2 \xi }\right]  \\
 &  \text{/} \left[4 \lambda _1^2 \lambda _2^6 \left(2 \lambda _1+\lambda _2\right)
   \left(\lambda _2^2-4 \lambda _1^2\right)^2\right]
 \end{cases}
 \end{align*}
   \begin{align*}
  \vdots \nonumber
%  \end{cases} 
 \end{align*}
therein  $u_-$ and $v_-$ are integration constants. Likewise, other higher order correction terms can be calculated using symbolic computations available in software packages Mathematica and Maple. Moreover, after calculating ten or higher order iteration terms the said software packages can easily provide the generating functions (or general term) of the iteration terms. Hither we have acquired the subsequent generating functions
 \begin{align}\label{msol2}
% \begin{cases}
 u(\xi, \epsilon)=&\left\{8 \lambda_1^2 \left(2 \lambda_1-\lambda_2\right) \left(2 \lambda_1+\lambda _2\right)^2 u_- e^{\lambda_1
   \xi } \left(2 \left(2 \lambda_1+\lambda_2\right) \lambda _2^2+\alpha_1     \left(\lambda _2-2 \lambda
   _1\right) \right. \right. \nonumber \\
  &  \left. \left. \times  v_- \epsilon  e^{\lambda _2 \xi }\right)\right\}  \text{/}Q(\xi,\epsilon), \nonumber \\
 v(\xi, \epsilon)=&4 \left(2 \lambda _1-\lambda _2\right) \left(2 \lambda _1+\lambda _2\right)^2 \epsilon  \left\{ \alpha_1    ^2
   \beta_1    ^2 \left(2 \lambda _1-\lambda _2\right) \lambda _2^4 u_-^4 v_- \epsilon ^4 e^{\left(4 \lambda
   _1+\lambda _2\right) \xi } \right. \nonumber \\
   &  \left. +16 \beta_1     \lambda _1^2  u_-^2 \epsilon  \left(2 \lambda _1+\lambda _2\right) e^{2
   \lambda _1 \xi } \left(4 \lambda _1^2 \left(2 \lambda _1+\lambda _2\right){}^2 \lambda _2^4+\alpha_1    ^2
   \lambda _1^2 v_-^2 \epsilon ^2  \right. \right. \nonumber \\
  &  \left. \left. \times \left(\lambda _2-2 \lambda _1\right)^2 e^{2 \lambda _2 \xi } -\alpha_1    
   \left(\lambda _2^3-4 \lambda _1^2 \lambda _2\right)^2 v_- \epsilon  e^{\lambda _2 \xi }\right)  -64 \lambda
   _1^4 \left(\lambda _2-2 \lambda _1\right) \right. \nonumber \\
  &  \left. \times \left(2 \lambda _1+\lambda _2\right)^4 \lambda _2^4 v_-
   e^{\lambda _2 \xi }\right\} \text{/} Q(\xi,\epsilon)^2
%  \end{cases} 
 \end{align}
 where
  \begin{align}\label{mgfun}
Q(\xi,\epsilon)=& \alpha_1     \beta_1     u_-^2 \epsilon ^2 e^{2 \lambda _1 \xi } \left\{ \alpha_1     \left(\lambda _2-2 \lambda _1\right)^2
   v_- \epsilon  e^{\lambda _2 \xi }-2 \lambda _2^2 \left(2 \lambda _1+\lambda _2\right)^2\right\}\nonumber \\
   & +8 \lambda
   _1^2 \left(2 \lambda _1-\lambda _2\right)  \left(2 \lambda _1+\lambda _2\right)^3  \left(2 \lambda_2^2-\alpha_1     v_- \epsilon  e^{\lambda _2 \xi }\right).
  \end{align}
 Ergo, the exact solution can be derived by setting $\epsilon=1$ in  (\ref{msol2})-(\ref{mgfun}) as
 \begin{align}\label{msol3}
% \begin{cases}
 u(\xi)=&\left\{8 \lambda_1^2 \left(2 \lambda_1-\lambda_2\right) \left(2 \lambda_1+\lambda _2\right)^2 u_- e^{\lambda_1
   \xi } \left(2 \left(2 \lambda_1+\lambda_2\right) \lambda _2^2+\alpha_1     \left(\lambda _2-2 \lambda
   _1\right) \right. \right.  \nonumber \\
   & \left. \left. \times v_-  e^{\lambda _2 \xi }\right)\right\}  \text{/}Q(\xi), \nonumber \\
 v(\xi)=&4 \left(2 \lambda _1-\lambda _2\right) \left(2 \lambda _1+\lambda _2\right)^2   \left\{ \alpha_1    ^2
   \beta_1    ^2 \left(2 \lambda _1-\lambda _2\right) \lambda _2^4 u_-^4 v_-  e^{\left(4 \lambda
   _1+\lambda _2\right) \xi }+16 \beta_1     \lambda _1^2    \right. \nonumber \\
   & \times  \left. u_-^2 \left(2 \lambda _1+\lambda _2\right) e^{2
   \lambda _1 \xi } \left(4 \lambda _1^2 \left(2 \lambda _1+\lambda _2\right)^2 \lambda _2^4+\alpha_1    ^2
   \lambda _1^2 \left(\lambda _2-2 \lambda _1\right)^2 v_-^2  e^{2 \lambda _2 \xi } \right. \right. \nonumber \\
  &  \left. \left. -\alpha_1    
   \left(\lambda _2^3-4 \lambda _1^2 \lambda _2\right)^2 v_-  e^{\lambda _2 \xi }\right)  -64 \lambda
   _1^4 \left(\lambda _2-2 \lambda _1\right) \left(2 \lambda _1+\lambda _2\right)^4 \lambda _2^4 v_-
   e^{\lambda _2 \xi }\right\} \nonumber \\
   &\text{/} Q(\xi)^2
%  \end{cases} 
 \end{align}
 where
  \begin{align}\label{mgfun2}
Q(\xi)=& \alpha_1     \beta_1     u_-^2  e^{2 \lambda _1 \xi } \left\{ \alpha_1     \left(\lambda _2-2 \lambda _1\right)^2
   v_-   e^{\lambda _2 \xi }-2 \lambda _2^2 \left(2 \lambda _1+\lambda _2\right)^2\right\}+8 \lambda
   _1^2 \nonumber \\
   & \times \left(2 \lambda _1-\lambda _2\right) \left(2 \lambda _1+\lambda _2\right)^3  \left(2 \lambda_2^2-\alpha_1     v_-   e^{\lambda _2 \xi }\right).
  \end{align}
  Due to the existence of many free parameters in the solution (\ref{msol3})-(\ref{mgfun2}), it is always not bounded ( for all values of these free parameters). To study physical relevant properties modelled by these equations, one always need to derive a bounded solution. To ensure the boundedness of this derived solution, below a theorem, have been presented.
%\subsubsection{Boundedness of solution (\ref{msol3})-(\ref{mgfun2})}  
\begin{theorem}\label{th1} The solution (\ref{msol3})-(\ref{mgfun2}) will be bounded if parameters $\lambda_1, \  \lambda_2, \ \alpha_1   , \ \beta_1   $ involved in the equation and integration constants $u_- , \ v_-$ satisfy any one of the  conditions given in the  table \ref{table0}.
\begin{table}[h]
\begin{center}
\caption{Boundedness conditions of theorem \ref{th1}. }\label{table0}
\begin{tabular}{c c c c c c c} \hline
Case& Con. No. &  $\lambda_2$ &$\alpha_1   $ & $\beta_1   $&  $u_-$ & $v_-$ \\ \hline \hline
    \multirow{4}{*}{$\lambda_1< \frac{\lambda_2}{2}<0$}&1(a)&$-$& $+$ &$ -$&$-$& $-$ \\
    &1(b)&$-$& $-$ &$ +$&$-$& $+$ \\
    &1(c)&$-$& $+$ &$ -$&$+$& $-$ \\
    &1(d)&$-$& $-$ &$ +$&$+$& $+$ \\ \hline
    
        \multirow{4}{*}{$\lambda_1>- \frac{\lambda_2}{2}>0$}&2(a)&$-$& $+$ &$ -$&$-$& $-$ \\
    &2(b)&$-$& $-$ &$ +$&$-$& $+$ \\
    &2(c)&$-$& $+$ &$ -$&$+$& $-$ \\
    &2(d)&$-$& $-$ &$ +$&$+$& $+$ \\ \hline
    
     \multirow{4}{*}{$\lambda_1<- \frac{\lambda_2}{2}<0$}&3(a)&$+$& $+$ &$ -$&$-$& $-$ \\
    &3(b)&$+$& $-$ &$ +$&$-$& $+$ \\
    &3(c)&$+$& $+$ &$ -$&$+$& $-$ \\
    &3(d)&$+$& $-$ &$ +$&$+$& $+$ \\ \hline
       \multirow{4}{*}{$\lambda_1> \frac{\lambda_2}{2}>0$}&4(a)&$+$& $+$ &$ -$&$-$& $-$ \\
    &4(b)&$+$& $-$ &$ +$&$-$& $+$ \\
    &4(c)&$+$& $+$ &$ -$&$+$& $-$ \\
    &4(d)&$+$& $-$ &$ +$&$+$& $+$ \\ \hline \hline
    
      \multirow{4}{*}{$ \frac{\lambda_2}{2}< \lambda_1<0$}&5(a)&$-$& $+$ &$ +$&$-$& $-$ \\
    &5(b)&$-$& $-$ &$ -$&$-$& $+$ \\
    &5(c)&$-$& $+$ &$ +$&$+$& $-$ \\
    &5(d)&$-$& $-$ &$ -$&$+$& $+$ \\ \hline
    
        \multirow{4}{*}{$0<\lambda_1<- \frac{\lambda_2}{2}$}&6(a)&$-$& $+$ &$ +$&$-$& $-$ \\
    &6(b)&$-$& $-$ &$ -$&$-$& $+$ \\
    &6(c)&$-$& $+$ &$ +$&$+$& $-$ \\
    &6(d)&$-$& $-$ &$ -$&$+$& $+$ \\ \hline

      \multirow{4}{*}{$ -\frac{\lambda_2}{2}< \lambda_1<0$}&7(a)&$+$& $+$ &$ +$&$-$& $-$ \\
    &7(b)&$+$& $-$ &$ -$&$-$& $+$ \\
    &7(c)&$+$& $+$ &$ +$&$+$& $-$ \\
    &7(d)&$+$& $-$ &$ -$&$+$& $+$ \\ \hline
    
         \multirow{4}{*}{$ 0 < \lambda_1<\frac{\lambda_2}{2}$}&8(a)&$+$& $+$ &$ +$&$-$& $-$ \\
    &8(b)&$+$& $-$ &$ -$&$-$& $+$ \\
    &8(c)&$+$& $+$ &$ +$&$+$& $-$ \\
    &8(d)&$+$& $-$ &$ -$&$+$& $+$ \\ \hline
\end{tabular}
\end{center}
\end{table}
\end{theorem}
\begin{proof}  Solution (\ref{msol3})-(\ref{mgfun2}) have a common factor in the denominator, which under transformation $e^{\xi}=Z$, reduces to a generalized Dirichlet polynomial \cite{jameson2006counting}
$$Pol(Z)= a_1 \  Z^{2 \lambda _1+\lambda _2}+a_2 \ Z^{2 \lambda_1} +a_3 \  Z^{\lambda _2}+a_4 ,   $$
where 
\begin{align*}
& a_1=\alpha_1    ^2 \beta_1     \left(\lambda _2-2 \lambda _1\right)^2 u_-^2 v_- ,\ \ a_2=-2 \alpha_1     \beta_1     \lambda _2^2 \left(2 \lambda _1+\lambda_2\right)^2 u_-^2 , \\
& a_3=-8 \alpha_1     \lambda _1^2 \left(2 \lambda _1-\lambda _2\right) \left(2 \lambda _1+\lambda _2\right)^3 v_- , \ \
 a_4=16 \lambda_1^2 \left(2 \lambda _1-\lambda _2\right) \lambda_2^2 \left(2 \lambda _1+\lambda _2\right)^3.
\end{align*}
 Solution (\ref{msol3})-(\ref{mgfun2}) is  unbounded if there exists at least one real positive root. So the boundedness of the solution is ensured by the conditions that the polynomial never have positive real root. Such conditions can be  provided by Descartes’ rule of
signs \cite{jameson2006counting}. Which states that the polynomial $Pol(Z)$ does not contain any positive real root if its coefficients never changes their signs. That leads us to the condition that either all $a_i>0\ (i=1,2,3,4)$ or all $a_i<0\ (i=1,2,3,4)$. Restrictions $a_i>0\ (i=1,2,3,4)$ yields the condition 1.(a)-4.(d) of the theorem, whereas remaining conditions provided by   $a_i<0\ (i=1,2,3,4)$. That completes the prove of the theorem.
\end{proof}
%%%%%%%%%%%%%%%%%%
Next, we are interested in establishing the conditions presented in the above theorem and study the features of the solution  (\ref{msol3})-(\ref{mgfun2}). For that, we have  taken particular values for different parameter  satisfying the conditions of the theorem in table \ref{table1} and utilising them to plot the solution shown in figure \ref{fig1}. From the 2D plots, it is clear that the solution is enriched with several one-hump, two-hump, W-shape, M-shape soliton like features.\\
\begin{table}[h]
\begin{center}
\caption{Particular values of parameters satisfying conditions presented in Theorem \ref{th1} used  in Fig. \ref{fig1}. }\label{table1}
\begin{tabular}{c c c c c c c} \hline
$\lambda_1$ &  $\lambda_2$ &$\alpha_1   $ & $\beta_1   $&  $u_-$ & $v_-$ & Figure \\ \hline \hline
    $-.3$&$-.5$& $1.9$ &$ -.5$&$-.5$& $-.6$ &1(a)\\
    $-.29$&$-.5$& $-.19$ &$ .5$&$-.5$& $.6$ &1(b)\\
  $-2.9$ &$-5$& $.19$ &$ -.5$&$.5$& $-.6$ &1(c) \\
    $-2.9$&$-5$& $-19$ &$ 5$&$5$& $6$ &1(d)\\ \hline
    $2.9$&$-5$& $19$ &$ -.15$&$-5$& $-2$&2(a) \\
    $2.9$&$-5$& $-19$ &$ .15$&$-5$& $6$&2(b) \\
    $.29$&$-.5$& $1.9$ &$ -5$&$5$& $-6$&2(c) \\
    $.29$&$-.5$& $-19$ &$ .15$&$5$& $6$ &2(d)\\ \hline
    
    $-.26$&$.5$& $19$ &$ -.15$&$-5$& $-6$&3(a) \\
    $-.28$&$.5$& $-19$ &$ 15$&$-5$& $6$ &3(b)\\
    $-.6$&$.5$& $1.9$ &$ -1.5$&$5$& $6$ &3(c)\\
    $-.06$&$.1$& $-1.9$ &$ 1.5$&$5$& $6$&3(d) \\ \hline
       $.27$&$.5$& $19$ &$ -.15$&$-5$& $-6$&4(a) \\
   $.27$&$.5$& $-19$ &$ .15$&$-5$& $6$ &4(b) \\
    $.4$&$.5$& $19$ &$ -.15$&$5$& $-1.6$&4(c) \\
    $.51$&$.5$& $-7$ &$ 1.5$&$5$& $6$&4(d) \\ \hline \hline
    
    $-.2$&$-.5$& $1.9$ &$ .5$&$-.5$& $-.6$ &5(a)\\
   $-.09$&$-.2$& $-1.9$ &$ -.5$&$-.5$& $.6$ &5(b) \\
    $-.05$&$-.2$& $.9$ &$.5$&$.5$& $-.6$&5(c) \\
    $-.09$&$-.2$& $-.9$ &$ -5$&$.5$& $.6$ &5(d)\\ \hline
    
        $.05$&$-.2$& $1.9$ &$ 1.5$&$-.01$& $-.6$&6(a) \\
   $.05$ &$-.2$& $-1.9$ &$ -1.5$&$-.1$& $.01$ &6(b)\\
    $.04$&$-.2$& $1.9$ &$ 1.5$&$.1$& $-.01$&6(c) \\
    $.06$&$-.2$& $-1.9$ &$ -1.5$&$.1$& $.01$&6(d) \\ \hline

   $-.09$&$.2$& $1.9$ &$ 1.5$&$-.1$& $-10.1$ &7(a)\\
    $-.03$&$.2$& $-1.9$ &$ -1.5$&$-.1$& $10.1$&7(b) \\
   $-.07$&$.2$& $1.2$ &$ 1.1$&$.3$& $-9$ &7(c) \\
    $-.11$&$.4$& $-1$ &$ -3$&$7$& $4$ &7(d)\\ \hline
    
   $.04$&$.1$& $.9$ &$ 1$&$-.7$& $-8$&8(a) \\
    $.03$&$.1$& $-.9$ &$ -1$&$-.7$& $7$&8(b) \\
    $.2$&$.5$& $8$ &$ .7$&$6$& $-7$&8(c) \\
    $.12$&$.4$& $-8$ &$ -.7$&$6$& $7$&8(d) \\ \hline
\end{tabular}
\end{center}
\end{table}
%%%%%%%%%%%%%%%%%%
\begin{figure*}[h]
\centering
\includegraphics[width=1\textwidth]{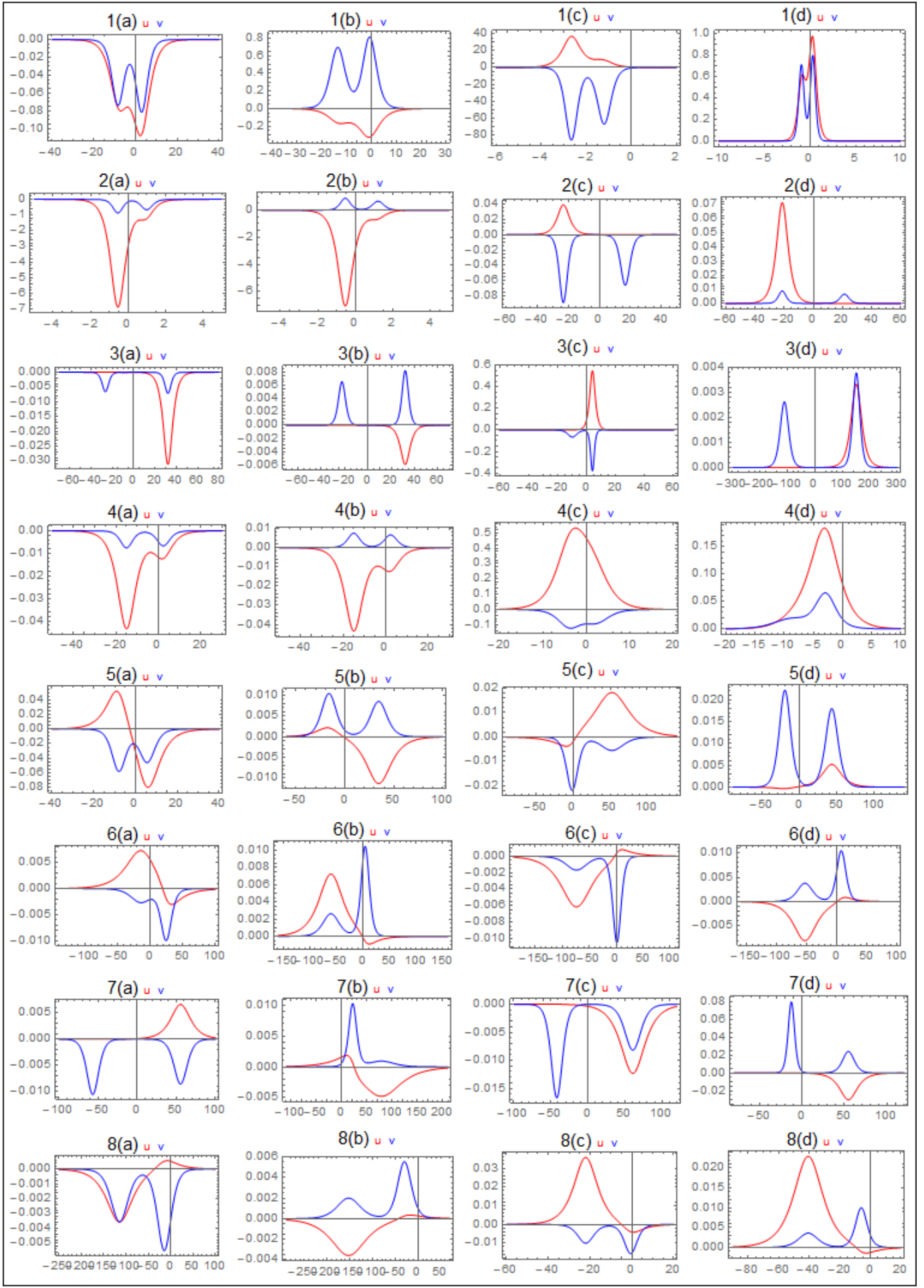}
\vskip 0.05in
\caption{Plots of the solution $u(\xi),\ v(\xi)$ for the conditions presented in Theorem \ref{th1}, using values of the parameters submitted  in
   Table \ref{table1}.}\label{fig1}       % Give a unique label
\end{figure*}
%%%%%%%%%%%%%%%%%%%
\textbf{Solution-I }\\
So in this case solution of (\ref{SBs}) can be obtained from (\ref{msol3})-(\ref{mgfun2}) with (\ref{SBsTr}) and (\ref{reimp}) in the form
 \begin{align}\label{SBSs1}
% \begin{cases}
 E_1(x, t)=&\left\{8 \lambda_1^2 \left(2 \lambda_1-\lambda_2\right) \left(2 \lambda_1+\lambda _2\right)^2 u_- e^{\lambda_1
   \xi } \left(2 \left(2 \lambda_1+\lambda_2\right) \lambda _2^2+\alpha_1     \left(\lambda _2-2 \lambda
   _1\right) \right. \right.  \nonumber \\
   & \times \left. \left.   v_-  e^{\lambda _2 \xi }\right)\right\}\text{/}Q(\xi)\ e^{i(k_1\ \frac{x^{\alpha}}{\alpha}+k_2\ \frac{t^{\beta}}{\beta}+c_0)}, \nonumber \\
 N_1(x, t)&=4 \left(2 \lambda _1-\lambda _2\right) \left(2 \lambda _1+\lambda _2\right)^2   \left\{ \alpha_1    ^2
   \beta_1    ^2 \left(2 \lambda _1-\lambda _2\right) \lambda _2^4 u_-^4 v_-  e^{\left(4 \lambda
   _1+\lambda _2\right) \xi } \right. \nonumber \\
   &   \left. +16 \beta_1     \lambda _1^2  u_-^2  \left(2 \lambda _1+\lambda _2\right) e^{2
   \lambda _1 \xi } \left(4 \lambda _1^2 \left(2 \lambda _1+\lambda _2\right)^2 \lambda _2^4+\alpha_1^2
   \lambda _1^2 \left(\lambda _2-2 \lambda _1\right)^2 \right. \right. \nonumber \\
  &  \left. \left. \times v_-^2  e^{2 \lambda _2 \xi }  -\alpha_1    
   \left(\lambda _2^3-4 \lambda _1^2 \lambda _2\right)^2 v_-  e^{\lambda _2 \xi }\right)  -64 \lambda
   _1^4 \left(\lambda _2-2 \lambda _1\right) \left(2 \lambda _1+\lambda _2\right)^4 \right. \nonumber \\
   & \left. \times \lambda _2^4 v_-  e^{\lambda _2 \xi }\right\} \text{/} Q(\xi)^2
%  \end{cases} 
 \end{align}
 where
  \begin{align}\label{SBSs1a}
Q(\xi)=& \alpha_1     \beta_1     u_-^2  e^{2 \lambda _1 \xi } \left\{ \alpha_1     \left(\lambda _2-2 \lambda _1\right)^2
   v_-   e^{\lambda _2 \xi }-2 \lambda _2^2 \left(2 \lambda _1+\lambda _2\right)^2\right\}+8 \lambda
   _1^2\nonumber \\
   & \times  \left(2 \lambda _1-\lambda _2\right) \left(2 \lambda _1+\lambda _2\right)^3  \left(2 \lambda_2^2-\alpha_1     v_-   e^{\lambda _2 \xi }\right),\ \ \  \xi=k_3 \ \frac{x^{\alpha}}{\alpha}+c\ \frac{t^{\beta}}{\beta}+\xi_0.
  \end{align}
This solution exists provided 
\begin{eqnarray}\label{res1}
\delta_3=-\frac{\delta_2 \mu _3}{3 \mu _2}, \ \ \frac{k_1 \left(\delta_1+\delta_2 k_1 \right)+k_2}{\delta_2 \ k_3^2} > 0, \ \ \text{and} \nonumber \\ -\frac{\delta_1^2+4 \delta _2 \delta _1 k_1+4 \delta _2^2 k_1^2+\mu _1}{\mu _2 \ k_3^2} > 0, 
\end{eqnarray}
which are yield from the integrability condition $\gamma_1    =3 \alpha_1   $ and the requirement that $\lambda_1, \ \lambda_2$ are real. The boundedness of solution (\ref{SBSs1})-(\ref{res1}) can be easily ensured by theorem \ref{th1} with the conditions (\ref{reimp}). Utilising one of such conditions the solution (\ref{SBSs1})-(\ref{res1}) has been plotted in figure \ref{fig2}. 
From the Fig. \ref{fig1}, on the (x,t)-plane, it is clear that absolute value of solution $E_1$ represents W-shaped soliton wave, real and imaginary parts of $E_1$  of the wave solution represent the Akhmediev breather (AB) wave, which can evolve periodically along  a certain angle with the t axis and component $N_1$ propagate with W-shaped soliton wave.
%%%%%%%%%%%%%%%%%%%%%%%%%%%%%%%%%%%%%%%%%%%%%%%%%%%%%%
 %%%%%%%%%%%%%%%%%%%%%%%%%%%%%%%
  \begin{figure}[h]
\captionsetup[subfigure]{labelformat=empty}
  \centering
\begin{subfigure}[t]{.35\linewidth}
    \centering
    \includegraphics[width=1\textwidth]{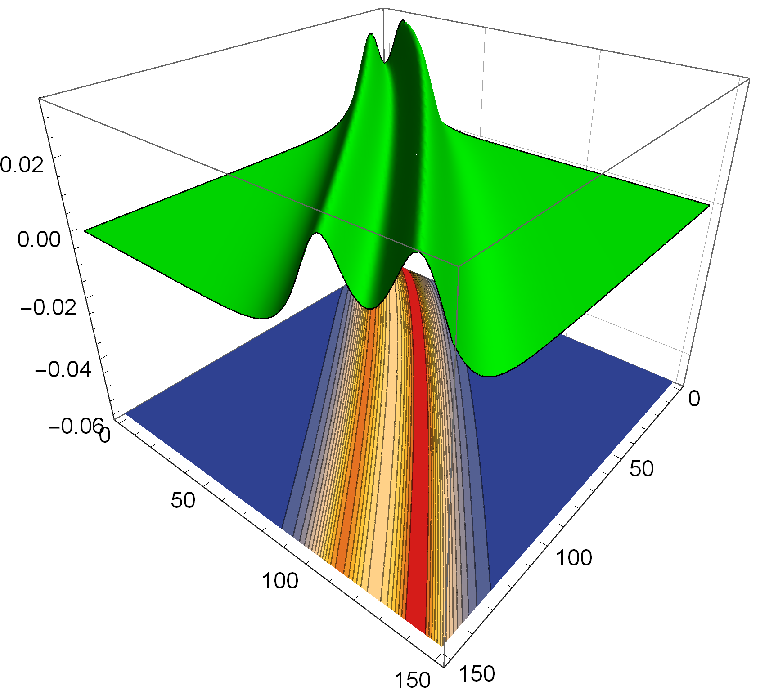}
    \caption{Abs($E_1$)}\label{fig:1b}
  \end{subfigure}% 
  \hspace*{1.9em}  
  \begin{subfigure}[t]{.35\linewidth}
    \centering
    \includegraphics[width=1\textwidth]{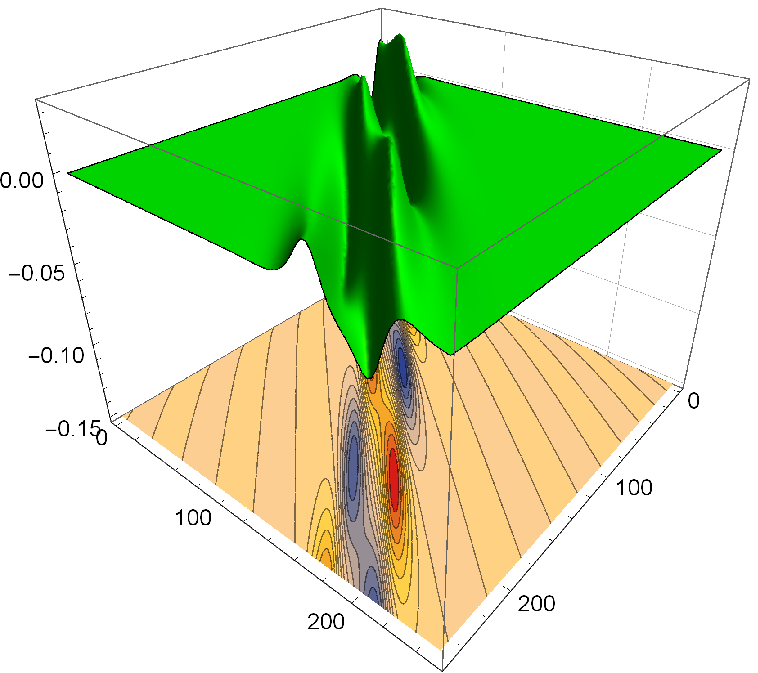}
     \caption{Re($E_1$)}\label{fig:1b}
  \end{subfigure} \\
   \vspace*{1em} 
  \begin{subfigure}[t]{.35\linewidth}
    \centering
    \includegraphics[width=1\textwidth]{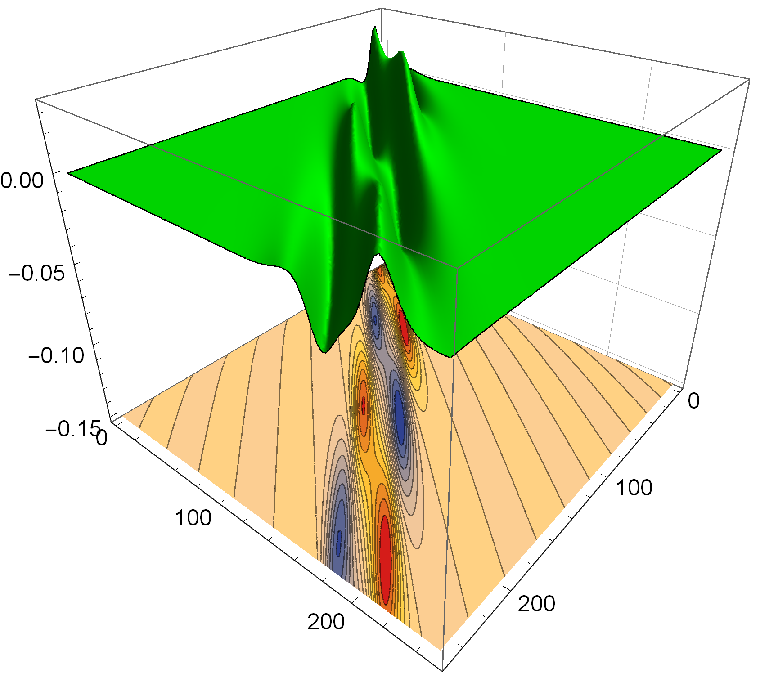}
    \caption{Im($E_1$)}\label{fig:1d}
  \end{subfigure}%
  \hspace*{1.9em}    
  \begin{subfigure}[t]{.35\linewidth}
    \centering
    \includegraphics[width=1\textwidth]{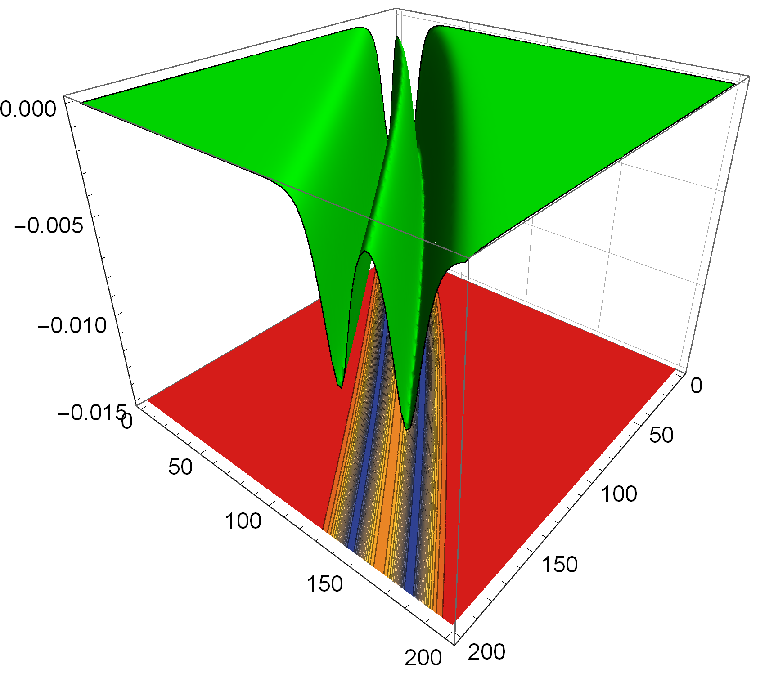}
    \caption{$N_1$}\label{fig:1n}
  \end{subfigure}  
  \caption{Plots of the solution (\ref{SBSs1})-(\ref{SBSs1a}) for  values of parameters values $k_1= 0.3, \ k_2= -0.1585, \ k_3= 2,\ \delta_1= 0.5,\ \delta_2= 0.5,\ \mu_1=- 0.515,\ \mu_2= -0.5,\ \mu_3=28.5,\ \mu_4= 0.075,\ \alpha=.8,\ \beta=.9,\ c_0=9, \ \xi_0=.4, \ u_-= -0.5,\ v_-=-0.6.$ }\label{fig2}
\end{figure}
%%%%%%%%%%%%%%%%%%%%%%%%%%%%%%%%
\subsection{Case-II \ $\alpha_1   =2\ \gamma_1    , \ \ \ \lambda_1=\lambda_2=\lambda$ (say)}
In this case, RCAM gives the following correction terms
 \begin{align*}
 %\begin{cases}
& \begin{cases}
  u_0(\xi )=u_-\ e^{-\lambda \ \xi },\\
  v_0(\xi )=v_- \ e^{-\lambda \ \xi },
 \end{cases} \\
&  \begin{cases}
  u_1(\xi )=\frac{2\ \gamma_1    \ u_-\ v_-\ e^{-2\ \lambda \ \xi }}{3\ \lambda^2},\\
  v_1(\xi )=\frac{e^{-2\ \lambda \ \xi } \left(\beta_1    \ u_-^2+\gamma_1    \  v_-^2\right)}{3 \ \lambda^2},
 \end{cases} \\
 & \begin{cases}
  u_2(\xi )=\frac{\gamma_1    \  u_-\ e^{-3 \lambda  \xi } \left(\beta_1   \  u_-^2+3 \gamma_1    \  v_-^2\right)}{12\ \lambda^4},\\
  v_2(\xi )=\frac{\gamma_1    \ v_- \ e^{-3 \lambda  \xi } \left(3\ \beta_1    \ u_-^2+\gamma_1     \ v_-^2\right)}{12 \lambda ^4},
 \end{cases}\\
&  \begin{cases}
  u_3(\xi )=\frac{2\ \gamma_1    ^2\ u_-\  v_-\ e^{-4 \lambda  \xi } \left(\beta_1   \  u_-^2+\gamma_1     \  v_-^2\right)}{27\ \lambda^6},\\
  v_3(\xi )=\frac{\gamma_1     \ e^{-4 \lambda  \xi } \left(\beta_1   ^2\ u_-^4+6\ \beta_1    \ \gamma_1     \ u_-^2 \ v_-^2+\gamma_1    ^2\ v_-^4\right)}{54\ \lambda ^6},
 \end{cases}
  \end{align*}
 \begin{align*}
& \ \ \ \ \vdots \\
& \begin{cases}
  u_m(\xi )= \frac{(m+1)e^{-\lambda  \xi }}{2^{m+1} 3^{m}  \sqrt{\beta_1   }} & \left[\left(\sqrt{\beta_1    } u_--\sqrt{\gamma_1     } v_-\right) \left[\frac{e^{-\lambda  \xi } \left(\gamma_1      v_--\sqrt{\beta_1    } \sqrt{\gamma_1     } u_-\right)}{\lambda ^2}\right]^m \right. \\
 & \left. +\left(\sqrt{\beta_1    } u_-+\sqrt{\gamma_1     } v_-\right) \left[\frac{e^{-\lambda  \xi } \left(\sqrt{\beta_1    } \sqrt{\gamma_1     } u_-+\gamma_1     
   v_-\right)}{\lambda ^2}\right]^m\right], \\
  v_m(\xi )= \frac{(m+1)e^{-\lambda  \xi }}{2^{m+1} 3^{m}  \sqrt{\gamma_1    }} & \left[\left(\sqrt{\gamma_1     } v_--\sqrt{\beta_1    } u_-\right) \left[\frac{e^{-\lambda  \xi } \left(\gamma_1      v_--\sqrt{\beta_1    } \sqrt{\gamma_1     } u_-\right)}{\lambda ^2}\right]^m \right. \\
& \left. +\left(\sqrt{\beta_1    } u_-+\sqrt{\gamma_1     } v_-\right) \left[\frac{e^{-\lambda  \xi } \left(\sqrt{\beta_1    } \sqrt{\gamma_1     } u_-+\gamma_1     
   v_-\right)}{\lambda ^2}\right]^m\right],
 \end{cases}\\
& \ \ \ \ \vdots
 % \end{cases} 
 \end{align*}
where $u_-$ and $v_-$ are integration constants. Summing the above series terms one can obtain close form solution of (\ref{meq1}) in the form
 \begin{eqnarray}\label{sol2}
 \begin{cases}
 u(\xi)=\frac{36 \lambda^4 u_- e^{\lambda  \xi } \left(36 \lambda ^4 e^{2 \lambda  \xi }+\beta_1     \gamma_1      u_-^2-\gamma_1     ^2 v_-^2\right)}{\left(36 \lambda ^4 e^{2 \lambda  \xi }-\beta_1    \gamma_1      u_-^2+\gamma_1     ^2 v_-^2-12 \gamma_1      \lambda ^2 v_- e^{\lambda  \xi }\right)^2},\\
 v(\xi)= \frac{36 \lambda ^4 e^{\lambda  \xi } \left(12 \beta_1     \lambda ^2 u_-^2 e^{\lambda  \xi }-\beta_1     \gamma_1      u_-^2 v_-+\gamma_1     ^2 v_-^3-12 \gamma_1      \lambda ^2 v_-^2 e^{\lambda  \xi}+36 \lambda ^4 v_- e^{2 \lambda  \xi }\right)}{\left(36 \lambda ^4 e^{2 \lambda  \xi }-\beta_1     \gamma_1      u_-^2+\gamma_1     ^2 v_-^2-12 \gamma_1      \lambda ^2 v_- e^{\lambda  \xi
   }\right)^2}.
  \end{cases} 
 \end{eqnarray}
In the following, a theorem has been proposed and proved to ensure the boundedness of this derived solution.
 \begin{theorem}\label{th2} The solution (\ref{sol2}) will be bounded if real parameters $\lambda, \  \beta_1   ,\  \ \gamma_1    $ present in the equation and integration constants $v_-, \ u_-$ involved in solution satisfies any one of the following conditions
\begin{description}
\item{Ia.} \ $\  \beta_1   >0\ \& \ \gamma_1     <0$
\item{Ib.} \ $\  \beta_1   <0\  \& \ \gamma_1     >0$
\item{IIa.} \ $\ \beta_1    >0,\ \gamma_1    >0,\ v_- <0    \  \&  -\sqrt{\frac{\gamma_1      v_-^2}{\beta_1    }}<u_-<\sqrt{\frac{\gamma_1     v_-^2}{\beta_1    }}$
\item{IIb.} \ $\ \beta_1    <0,\ \gamma_1    <0,\ v_- >0    \  \&  -\sqrt{\frac{\gamma_1      v_-^2}{\beta_1    }}<u_-<\sqrt{\frac{\gamma_1     v_-^2}{\beta_1    }}$
\end{description}
\end{theorem}
\begin{proof} The solution (\ref{msol3}) components  have common denominator given by the cubic polynomial
$$Pol(Z)= -\beta_1     \gamma_1      u_-^2+\gamma_1     ^2 v_-^2-12 \gamma_1      \lambda ^2 v_- Z+36 \lambda ^4 Z^2, $$
where $Z=\text{exp}(\lambda\ \xi)$. Solution (\ref{sol2}) is bounded if the roots of the polynomial are  either negative real or complex. In renaming cases they are unbounded. So boundedness of solution demands that roots of $Pol(Z)$ have to be negative real, complex, or both. Roots of the above polynomial are given by
\begin{eqnarray*}
Z_{\pm}=\frac{\gamma_1     v_-\pm \sqrt{\beta_1    } \sqrt{\gamma_1     } u_-}{6 \lambda ^2}
\end{eqnarray*}
Roots $Z_{\pm}$ will be complex if $\beta_1    \gamma_1     <0$, which provides boundedness conditions Ia \& Ib of the theorem. Remaining boundedness conditions can be obtained when  $Z_{\pm}$ are negative real. Now  $Z_{-}$ will be negative real if parameters satisfy the relations
\begin{description}
\item{I.} $\  \beta_1   >0, \ \gamma_1     >0\ \&  \ (\gamma_1     v_-- \sqrt{\beta_1    } \sqrt{\gamma_1     } u_-)<0 $
\item{II.} $\  \beta_1   <0, \ \gamma_1     <0\ \&  \ (\gamma_1     v_-- \sqrt{\beta_1    } \sqrt{\gamma_1     } u_-)<0. $
\end{description}
Which further can be put in the forms
\begin{description}
\item{I.a.} \ $\beta_1    >0,\ \gamma_1     >0,\ v_-\leq 0,\ \& \ u_->-\sqrt{\frac{\gamma_1      v_-^2}{\beta_1    }} $
\item{I.b.} \ $\beta_1    >0,\ \gamma_1     >0,\ v_-> 0,\ \& \ u_->\sqrt{\frac{\gamma_1      v_-^2}{\beta_1    }} $
\item{II.a.} \ $\beta_1    <0,\ \gamma_1     <0,\ v_-> 0,\ \& \ u_->-\sqrt{\frac{\gamma_1      v_-^2}{\beta_1    }} $
\item{II.b.} \ $\beta_1    <0,\ \gamma_1     <0,\ v_-\leq 0,\ \& \ u_->\sqrt{\frac{\gamma_1      v_-^2}{\beta_1    }}. $
\end{description}
Now  $Z_{+}$ will be negative real if parameters satisfy the relations
\begin{description}
\item{1.} $\  \beta_1   >0, \ \gamma_1     >0\ \&  \ (\gamma_1     v_-+ \sqrt{\beta_1    } \sqrt{\gamma_1     } u_-)<0 $
\item{2.} $\  \beta_1   <0, \ \gamma_1     <0\ \&  \ (\gamma_1     v_-+ \sqrt{\beta_1    } \sqrt{\gamma_1     } u_-)<0 $
\end{description}
Which further can be recast in the forms
\begin{description}
\item{1.a.} $\  \beta_1   >0, \ \gamma_1     >0,\ v_- >0\  \&  \ u_-<-\sqrt{\frac{\gamma_1      v_-^2}{\beta_1    }} $
\item{1.b.} $\  \beta_1   >0, \ \gamma_1     >0,\ v_- \leq 0\  \ \&  \ u_-<\sqrt{\frac{\gamma_1      v_-^2}{\beta_1    }} $
\item{2.a.} $\beta_1    <0,\ \gamma_1     <0,v_-\leq 0, \ \& \ u_-< -\sqrt{\frac{\gamma_1      v_-^2}{\beta_1    }}$
\item{2.b.} $\beta_1    <0,\ \gamma_1     <0,\ v_->0,  \ \& \ u_- < \sqrt{\frac{\gamma_1      v_-^2}{\beta_1    }}.$
\end{description}
Now combining negative real  $Z_{\pm}$ conditions, we get the valid two common  regions summarised in conditions IIa. \& IIb. of the theorem.  
  Hence the theorem is proved.
\end{proof}
\begin{table}[h]
\begin{center}
\caption{Particular values of parameters satisfying conditions presented in Theorem \ref{th2} \& \ref{th3}, used  in Fig. \ref{fig3}. }\label{table2}
\begin{tabular}{c c c c c c c} \hline 
 $\lambda$ &$\alpha_1   $ & $\beta_1   $& $\gamma_1    $ & $u_-$ & $v_-$&Fig. \\ \hline \hline
$.3$&&$30$&$-20$&$2$ &$-15$&Ia\\ 
$.3$&&$-30$&$20$&$-2$& $15$&Ib \\ 
$.3$&&$30$&$20$&$2.44 $& -3&IIa\\ 
$.3$&&$-30$&$-20$&$18.7$& -3&IIb\\ \hline \hline

$.8$&$.1$ &$30$&$-.01$&$-2$ & &IIIa\\ 
$.8$&$1$&$-3$&$3$&$-2$& &IIIb \\ 
$1$&$-.1$&$.3$&$-3$&$2$& &IIIc\\ 
$1$&$-1$&$-3$&$3$&$2$& &IIId\\ \hline
\end{tabular}
\end{center}
\end{table}
\begin{figure}[h]
    \centering
    \includegraphics[width=1\textwidth]{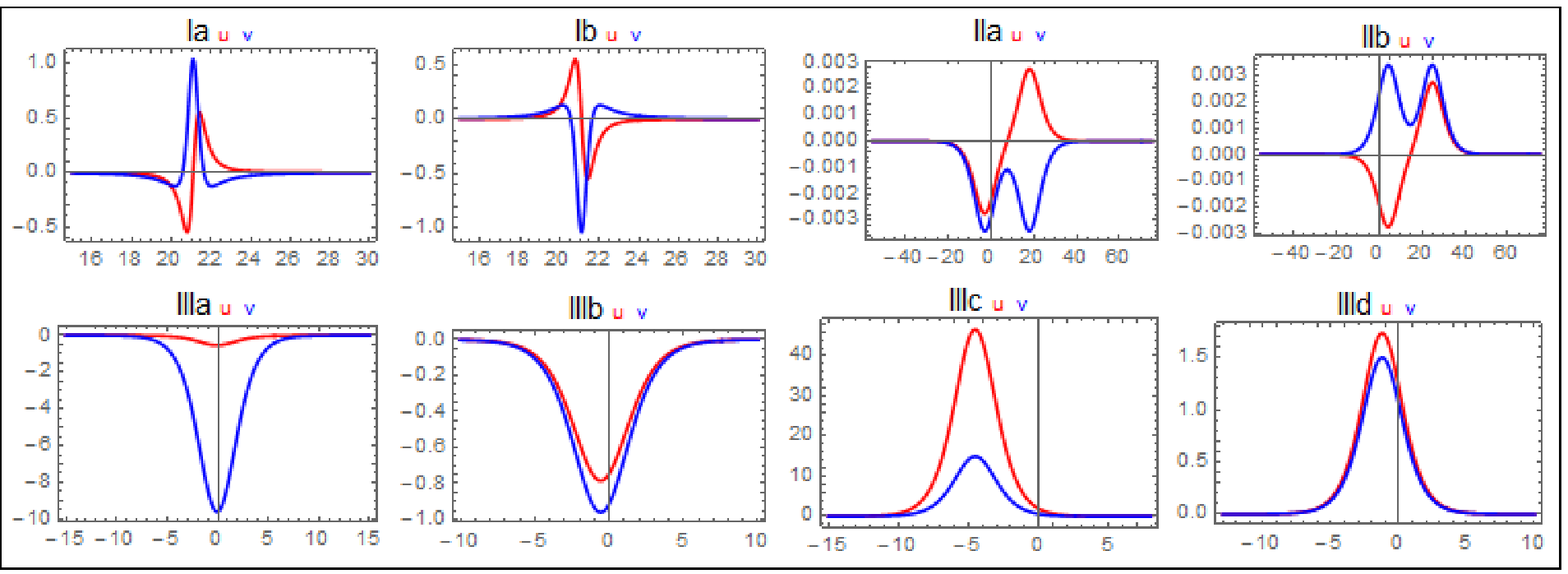}  
  \caption{Plot of solutions (\ref{sol2}) and (\ref{sol3}) for the conditions presented in Theorem \ref{th2} \& \ref{th3}, using values of parameters given in Table \ref{table2}. }\label{fig3}
\end{figure}
Next to establishing the conditions presented in the above theorem and study the features of the solution  (\ref{sol2}), we have  taken particular values of different parameters  satisfying the conditions of the theorem in table \ref{table2} and utilising them to plot the solution  in figure \ref{fig3}. From the 2D plots, it is clear that the solution is enriched with several  two-hump, W-shape, M-shape soliton-like profile.\\
\textbf{Solution-II }\\
So in this case solution of (\ref{SBs}) can be obtained from (\ref{sol2}) with (\ref{SBsTr}) and (\ref{reimp}) in the form
 \begin{align}\label{SBSs2}
% \begin{cases}
 E_2(x, t)=&\frac{36 \lambda^4 u_- e^{\lambda  \xi } \left(36 \lambda ^4 e^{2 \lambda  \xi }+\beta_1     \gamma_1      u_-^2-\gamma_1     ^2 v_-^2\right)}{\left(36 \lambda ^4 e^{2 \lambda  \xi }-\beta_1    \gamma_1      u_-^2+\gamma_1     ^2 v_-^2-12 \gamma_1      \lambda ^2 v_- e^{\lambda  \xi }\right)^2}\ e^{i(k_1\ \frac{x^{\alpha}}{\alpha}+k_2\ \frac{t^{\beta}}{\beta}+c_0)}, \nonumber \\
 N_2(x, t)=& 36 \lambda ^4 e^{\lambda  \xi } \left[12 \beta_1     \lambda ^2 u_-^2 e^{\lambda  \xi }-\beta_1     \gamma_1      u_-^2 v_-+\gamma_1     ^2 v_-^3-12 \gamma_1      \lambda ^2 v_-^2 e^{\lambda  \xi} \right. \nonumber  \\
 & \left. +36 \lambda ^4 v_- e^{2 \lambda  \xi }\right] \text{/} \left[36 \lambda ^4 e^{2 \lambda  \xi }-\beta_1     \gamma_1      u_-^2+\gamma_1     ^2 v_-^2-12 \gamma_1      \lambda ^2 v_- e^{\lambda  \xi
   }\right]^2,
%  \end{cases} 
 \end{align}
where  $ \xi=k_3 \ \frac{x^{\alpha}}{\alpha}+c\ \frac{t^{\beta}}{\beta}+\xi_0.$ This solution exist provided 
 \begin{align}\label{SBSs2a}
& \delta_3=-\frac{2 \delta _2 \mu _3}{\mu _2}, \ \  k_2=-\frac{\delta _2 \delta _1^2+\delta _1 k_1 \left(4 \delta _2^2+\mu _2\right)+\delta _2 \left(k_1^2 \left(4 \delta _2^2+\mu _2\right)+\mu _1\right)}{\mu _2},\ \  \text{and} \nonumber \\ 
& -\frac{\delta _1^2+4 \delta _2 \delta _1 k_1+4 \delta _2^2 k_1^2+\mu _1}{\mu _2\ k_3^2}>0,
 \end{align}
which are derived using the integrability  condition $\alpha_1   =2\gamma_1    , \ \  \lambda_1=\lambda_2=\lambda$ and demand that $\lambda$ is real.

The boundedness of solution (\ref{SBSs2})-(\ref{SBSs2a}) can be easily ensured by theorem \ref{th2} with the conditions (\ref{reimp}). Using one among those conditions solution (\ref{SBSs2})-(\ref{SBSs2a}) is plotted in figure \ref{fig3}. 
 Fig. \ref{fig3} shows that the absolute value of solution $E_2$ have M-shaped soliton profile, whereas real and imaginary parts of $E_2$  of the wave solution constitute the Akhmediev breather (AB) wave, the wave is not the space-periodic breather but the
time-periodic breather and component $N_2$ propagate with keeping M-shaped form soliton wave.
 %%%%%%%%%%%%%%%%%%%%%%%%%%%%%%%%%%%%%%%%%%%%%%%%%
\begin{figure}[h]
\captionsetup[subfigure]{labelformat=empty}
\centering
\begin{subfigure}[b]{.35\linewidth}
    \centering
    \includegraphics[width=1\textwidth]{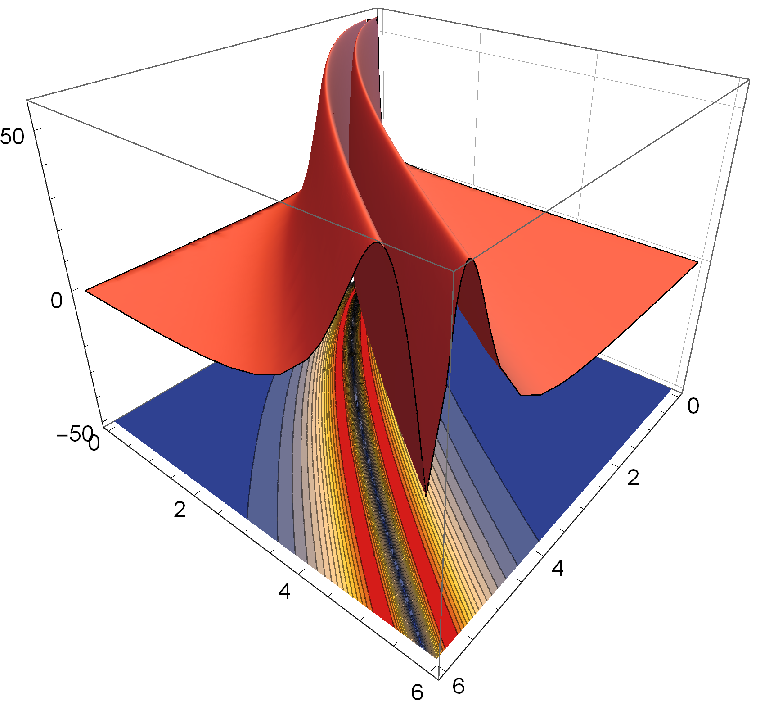}
    \caption{Abs($E_2$)}\label{fig:1b}
  \end{subfigure}%  
   \hspace*{1.9em}   
  \begin{subfigure}[b]{.35\linewidth}
    \centering
    \includegraphics[width=1\textwidth]{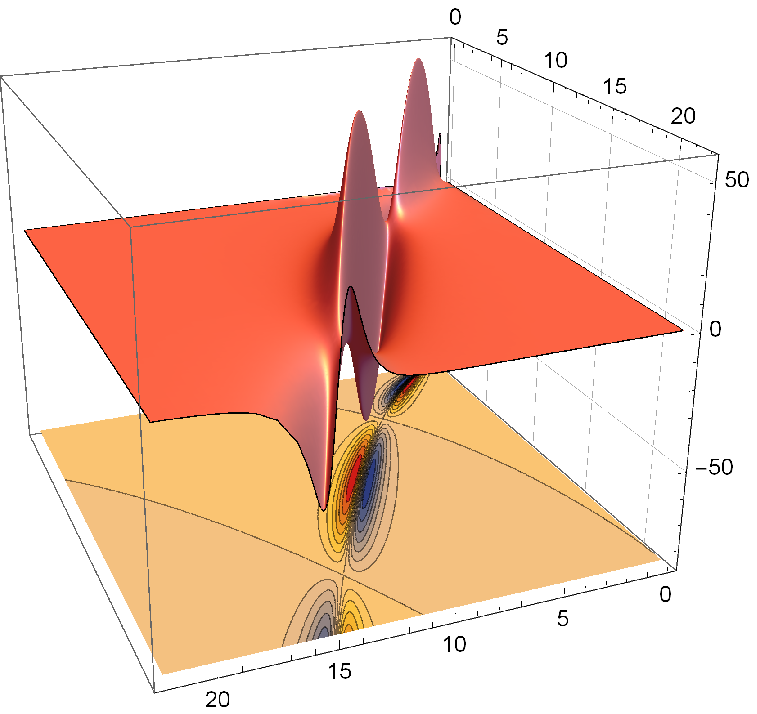}
    \caption{Re($E_2$)}\label{fig:1b}
  \end{subfigure} \\
   \vspace*{1em}  
  \begin{subfigure}[b]{.35\linewidth}
    \centering
    \includegraphics[width=1\textwidth]{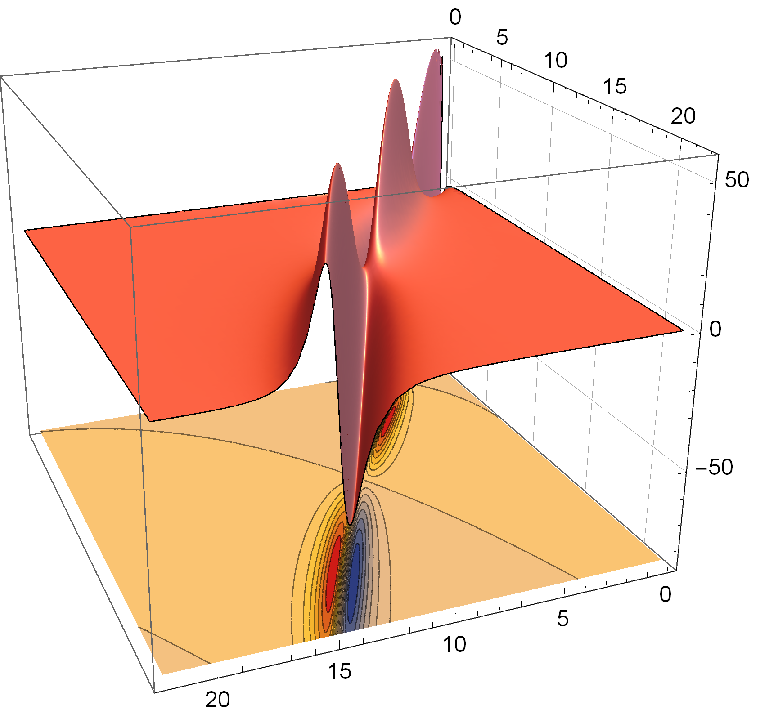}
    \caption{Im($E_2$)}\label{fig:1d}
  \end{subfigure}%
   \hspace*{1.9em}     
  \begin{subfigure}[b]{.35\linewidth}
    \centering
    \includegraphics[width=1\textwidth]{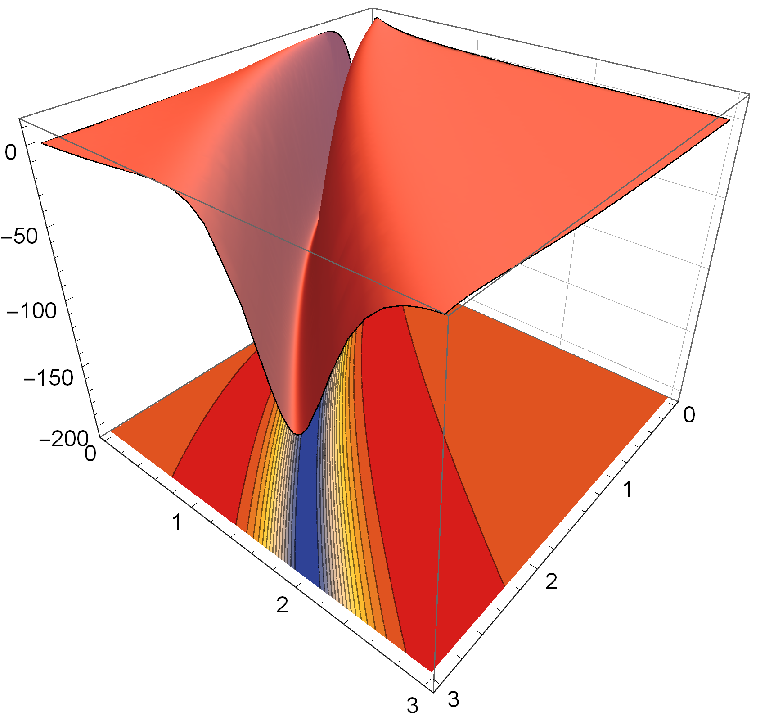}
    \caption{$N_2$}\label{fig:1n}
  \end{subfigure}  
  \caption{Plot of the solution (\ref{SBSs2})-(\ref{SBSs2a}) for  values of parameters values $k_1= 0.4, \ k_3= 1.4, \ \delta_1= 0.5,\ \delta_2= 0.6,\ \mu_1= 0.2,\ \mu_2= -1,\ \mu_3=0.4,\ \mu_4= 0.5,\ u_-= 0.4,\ v_-=0.99, \ \alpha=.7,\ \beta=.5,\ c_0=1, \ \xi_0=-4.$ }\label{fig4}
\end{figure}
%%%%%%%%%%%%%%%%%%%%%%%%%%%%%%%%%%%%%%%%%%%%%%%%%
In the previous cases, we have derived exact solutions of SBS by eliminating one or more free parameters using  compatibility conditions of the Painlev$\acute{\text{e}}$  test. Then it is of natural curiosity to explore whether it is possible to obtain an exact solution of the equations without using compatibility conditions (nonintegrable case) containing all parameters involved in the equation? To answer that here we consider a nonintegrable case $\lambda_1=\lambda_2.$ To solve the considered system of equations rather eliminating any other additional parameters, we eliminate one integration constant.
\subsection{Case-III \ $v_-=\frac{\sqrt{\beta_1    } u_-}{\sqrt{\alpha_1    -\gamma_1     }}, \ \  \lambda_1=\lambda_2=\lambda$ (say)}
In this case, RCAM gives the following correction terms
 \begin{align*}
% \begin{cases}
& \begin{cases}
  u_0(\xi )=u_-\ e^{-\lambda \ \xi },\\
  v_0(\xi )= \frac{\sqrt{\beta_1    } \ u_- \ e^{-\lambda \ \xi }}{\sqrt{\alpha_1    -\gamma_1     }},
 \end{cases}   
 \end{align*}
   \begin{align*}
&  \begin{cases}
  u_1(\xi )= \frac{\alpha_1   \  \sqrt{\beta_1    }\ u_-^2 \ e^{-2 \lambda  \xi }}{3 \lambda^2\ \sqrt{\alpha_1    -\gamma_1     }},\\
  v_1(\xi )=\frac{\alpha_1   \  \beta_1    \ u_-^2\ e^{-2 \lambda \ \xi }}{3 \lambda^2 (\alpha_1    -\gamma_1     )},
 \end{cases}  
  \end{align*}
   \begin{align*}
 & \begin{cases}
  u_2(\xi )=\frac{\alpha_1   ^2\ \beta_1    \ u_-^3\ e^{-3 \lambda \ \xi }}{12 \lambda^4 (\alpha_1    -\gamma_1     )},\\
  v_2(\xi )= \frac{\alpha_1   ^2\ \beta_1   ^{\frac{3}{2}}\ u_-^3\ e^{-3 \lambda\  \xi }}{12 \ \lambda^4 (\alpha_1    -\gamma_1     )^{\frac{3}{2}}},
 \end{cases}\\
&  \begin{cases}
  u_3(\xi )= \frac{\alpha_1   ^3 \ \beta_1   ^{\frac{3}{2}}\ u_-^4 e^{-4 \lambda \ \xi }}{54 \lambda^6 (\alpha_1    -\gamma_1     )^{\frac{3}{2}}},\\
  v_3(\xi )=\frac{\alpha_1   ^3\ \beta_1   ^2\ u_-^4\ e^{-4 \lambda \ \xi }}{54 \lambda^6 (\alpha_1    -\gamma_1     )^2},
 \end{cases}
  \end{align*}
   \begin{align*}
& \ \ \ \ \vdots \\
& \begin{cases}
  u_m(\xi )=\frac{ (m+1)\ u_- \ e^{-\lambda  \xi }}{6^{m}} \left[ \frac{\alpha_1     \sqrt{\beta_1    } u_- e^{-\lambda  \xi }}{\lambda^2 \sqrt{\alpha_1    -\gamma_1     }}\right]^m, \\
  v_m(\xi )=\frac{\lambda^2 \ (m+1) }{6^{m}\ \alpha_1    } \left[ \frac{\alpha_1     \sqrt{\beta_1    } u_- e^{-\lambda  \xi }}{\lambda ^2 \sqrt{\alpha_1    -\gamma_1    }}\right]^{m+1},
 \end{cases}\\
&\ \ \ \ \vdots
%  \end{cases} 
 \end{align*}
where $u_-$ and $v_-$ are integration constants. Summing the above series terms one can derive the close form solution of (\ref{meq1}) in the form
 \begin{eqnarray}\label{sol3}
 \begin{cases}
 u(\xi)=\frac{36 u_- \lambda^4  (\alpha_1    -\gamma_1    )^2 e^{\lambda  \xi }}{\left(6 \lambda^2 (\alpha_1    -\gamma_1    ) e^{\lambda  \xi }-u_- \alpha_1     \sqrt{\beta_1    }  \sqrt{\alpha_1    -\gamma_1     }\right)^2},\\
 v(\xi)= \frac{36 u_- \lambda^4 \sqrt{\beta_1    }  (\alpha_1    -\gamma_1    )^{\frac{3}{2}} e^{\lambda  \xi }}{\left(6 \lambda^2 (\alpha_1    -\gamma_1    ) e^{\lambda  \xi }- u_- \alpha_1     \sqrt{\beta_1    } \sqrt{\alpha_1    -\gamma_1     }\right)^2}.
  \end{cases} 
 \end{eqnarray}
 To make the solution (\ref{sol3}) physically relevant below we present a theorem to derive its bounded cases.
 \begin{theorem}\label{th3} The solution (\ref{sol3}) will be bounded if parameters $\alpha_1   ,\  \beta_1   ,\  \gamma_1    $ involved in the equation and integration constant $u_-$ involve in solution satisfies any one of the following conditions
\begin{description}
\item{IIIa} \ $\alpha_1   >0,\  \beta_1   >0,\ \ u_-<0\ \& \ \alpha_1    > \gamma_1    $
\item{IIIb} \ $\alpha_1   >0,\  \beta_1   <0,\ \ u_-<0\ \& \ \alpha_1    < \gamma_1    $
\item{IIIc} \ $\alpha_1   <0,\  \beta_1   >0,\ \ u_->0 \ \& \ \alpha_1    > \gamma_1    $
\item{IIId} \ $\alpha_1   <0,\  \beta_1   <0,\ \ u_->0\ \& \ \alpha_1    < \gamma_1    $
\end{description}
\end{theorem}
\begin{proof} The components $u(\xi)$ and $v(\xi)$ of the solution (\ref{sol3}) have a common denominator given by the quartic polynomial
$$\left(6 \lambda^2 (\alpha_1    -\gamma_1    ) e^{\lambda  \xi }- u_- \alpha_1     \sqrt{\beta_1    } \sqrt{\alpha_1    -\gamma_1     }\right)^2. $$
The  repeated roots of the polynomial given by
$$\xi=\frac{1}{\lambda} log\left[  \frac{\alpha_1    \ u_- \sqrt{\beta_1   }}{6  \lambda^2 \sqrt{\alpha_1   -\gamma_1    }} \right].$$
The positive real values of these roots make denominator  zero that leads to an unbounded solution. So the requirement of boundedness of solution suggests that the argument of the $log$ has to be complex or negative. The complex case is not admissible  here because it makes solutions complex, so the real bounded solutions given  by the conditions among parameters 
$$\beta_1    ( \alpha_1    - \gamma_1     )>0 \, \, \text{and} \, \, \alpha_1   \ u_- <0 .$$
These restrictions can be split to derive the conditions of the theorem.
\end{proof}
%%%%%%%%%%%%%%%%%%%%%%%%
Few particular values of the free parameters which  satisfy the conditions of the theorem presented  in table \ref{table2} and using them solution have been  plotted  in figure \ref{fig3}. The 2D plots ensure that solution always has a one-hump soliton-like profile.\\
\textbf{Solution-III }\\
So the solution of (\ref{SBs}) can be obtained from (\ref{sol3}) with (\ref{SBsTr}) and (\ref{reimp}) in the form
 \begin{align}\label{SBSs3}
% \begin{cases}
 E_3(x, t)=& \frac{36 u_- \lambda^4  (\alpha_1    -\gamma_1    )^2 e^{\lambda  \xi }}{\left(6 \lambda^2 (\alpha_1    -\gamma_1    ) e^{\lambda  \xi }-u_- \alpha_1     \sqrt{\beta_1    }  \sqrt{\alpha_1    -\gamma_1     }\right)^2}\ e^{i(k_1\ \frac{x^{\alpha}}{\alpha}+k_2\ \frac{t^{\beta}}{\beta}+c_0)}, \nonumber \\
 N_3(x, t)=& \frac{36 u_- \lambda^4 \sqrt{\beta_1    }  (\alpha_1    -\gamma_1    )^{\frac{3}{2}} e^{\lambda  \xi }}{\left(6 \lambda^2 (\alpha_1    -\gamma_1    ) e^{\lambda  \xi }- u_- \alpha_1     \sqrt{\beta_1    } \sqrt{\alpha_1    -\gamma_1     }\right)^2}, \ \ \ \xi=k_3 \ \frac{x^{\alpha}}{\alpha}+c\ \frac{t^{\beta}}{\beta}+\xi_0,
%  \end{cases} 
 \end{align}
provided 
 \begin{align}\label{SBSs3a}
  &k_2=-\frac{\delta _2 \delta _1^2+\delta _1 k_1 \left(4 \delta _2^2+\mu _2\right)+\delta _2 \left(k_1^2 \left(4 \delta _2^2+\mu _2\right)+\mu _1\right)}{\mu _2}, \ \ \text{and} \nonumber \\
   & -\frac{\delta _1^2+4 \delta _2 \delta _1 k_1+4 \delta _2^2 k_1^2+\mu _1}{\mu _2\ k_3^2}>0,
 \end{align}
 which are given by condition $ \lambda_1=\lambda_2=\lambda$ and assertion that $\lambda$ is real.
 %%%%%%%%%%%%%%%%%%%%%%%%%%%%%%%
One can easily derive the boundedness conditions of solution (\ref{SBSs3})-(\ref{SBSs3a}) by using theorem \ref{th3} with the conditions (\ref{reimp}). Using one among those conditions graphical representations of the solution (\ref{SBSs3})-(\ref{SBSs3a}) are presented in figure \ref{fig3}.  The figure shows that the absolute value of solution $E_3$ has always soliton profile, whereas real and imaginary parts of $E_3$  of the wave solution constitute a space-periodic breather  and component $N_3$ propagate with maintaining soliton shaped form. 
 \begin{figure}[h]
\captionsetup[subfigure]{labelformat=empty}
\centering
   \begin{subfigure}[b]{.35\linewidth} 
    \centering
    \includegraphics[width=1\textwidth]{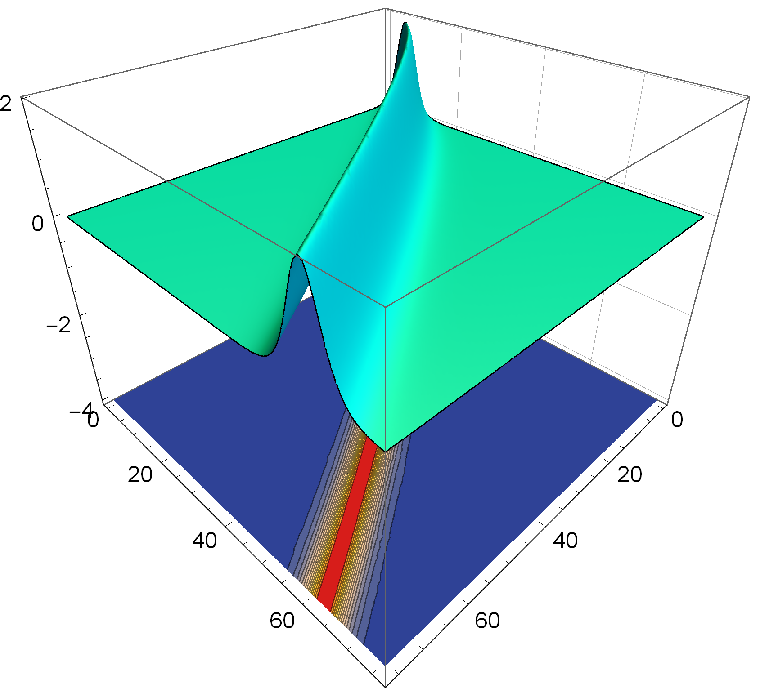}
    \caption{Abs($E_3$)}\label{fig:1b}
  \end{subfigure}% 
     \hspace*{1.9em}   
  \begin{subfigure}[b]{.35\linewidth}
    \centering
    \includegraphics[width=1\textwidth]{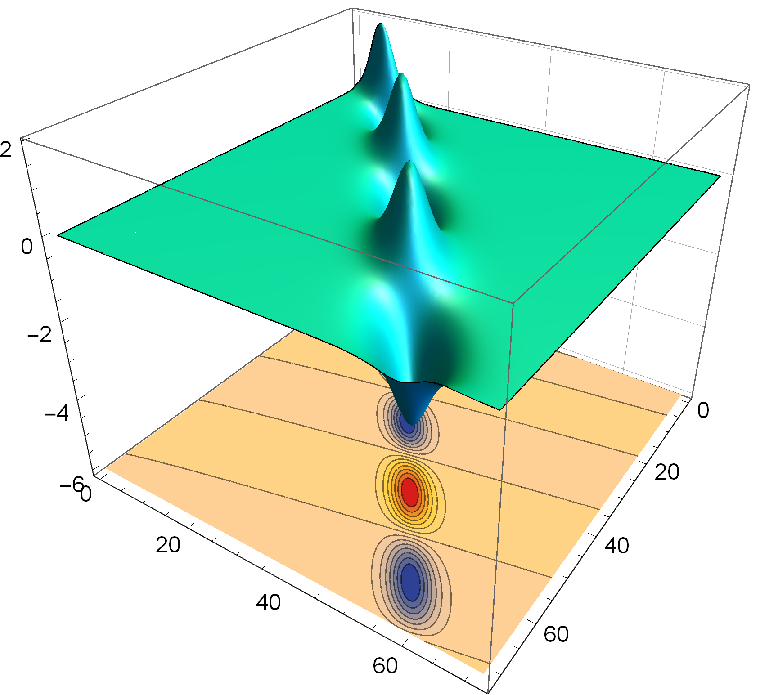}
    \caption{Re($E_3$)}\label{fig:1b}
  \end{subfigure}\\ 
     \vspace*{1em} 
  \begin{subfigure}[b]{.35\linewidth}
    \centering
    \includegraphics[width=1\textwidth]{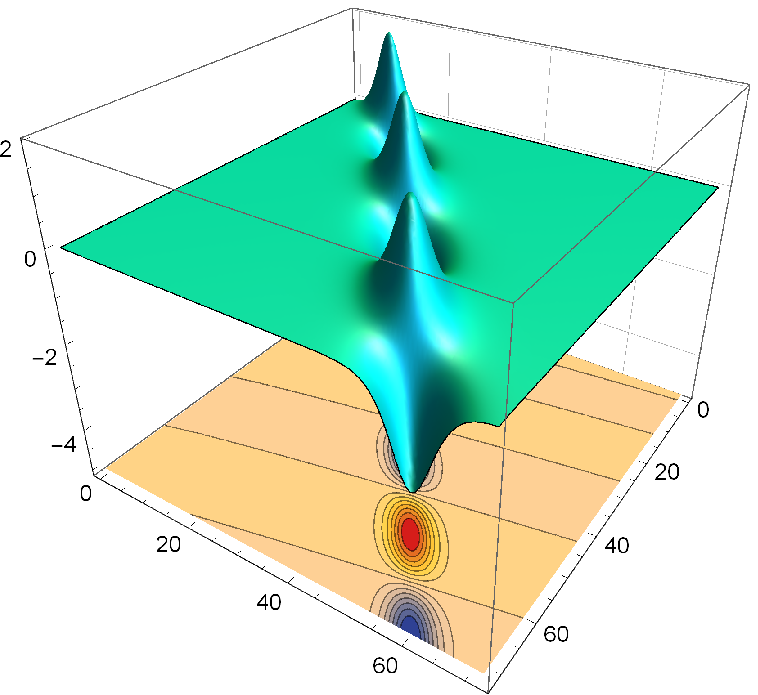}
    \caption{Im($E_3$)}\label{fig:1d}
  \end{subfigure}%  
     \hspace*{1.9em}  
  \begin{subfigure}[b]{.35\linewidth}
    \centering
    \includegraphics[width=1\textwidth]{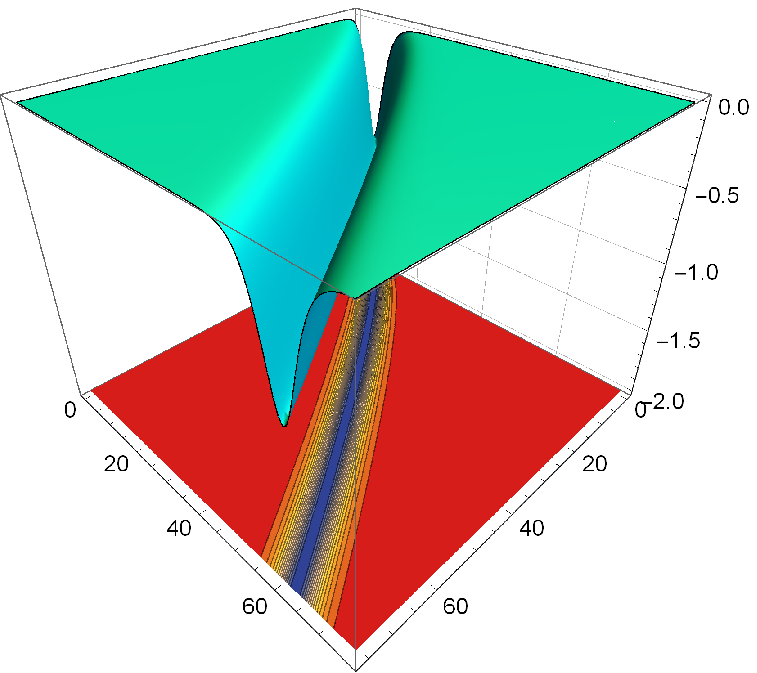}
    \caption{$N_3$}\label{fig:1n}
  \end{subfigure}%  
  \caption{Plot of solution (\ref{SBSs3}) for  values of parameters values $k_1= 0.4, \ k_3= 0.4,\ \delta_1= -1,\ \delta_2= 0.6,\ \delta_3= 0.6,\ \mu_1= 0.2,\ \mu_2= -2,\ \mu_3=-0.4,\ \mu_4=- 0.5, \ \alpha=.8,\ \beta=.9,\ c_0=9, \ \xi_0=3,\ u_-= -37.4.$ }\label{fig5}
\end{figure}
\section{Conclusions}\label{sec6}
In this work,  a few new exact solutions of three integrable/nonintegrable cases of  SBS with space-time conformable have been derived. Here we have applied a complex travelling wave transformation and a modified RCAM to solve the SBS. Painlev$\acute{\text{e}}$ test is used to identify the integrable cases of the stationary form of SBS. In general Painlev$\acute{\text{e}}$ analysis cannot handle the case when considered equation contains parameter coefficients. In this case, the arbitrariness of parameters leads to the symbolic resonances and one can not execute the remaining steps of Painlev$\acute{\text{e}}$ test. But here we have proposed an alternative technique to handle this limitation and derived integrable cases. 
We have classified all the bounded physically relevant cases of the solutions and presented them in three theorems. General theories of algebra  have been utilized to prove the theorems. In addition to that, all bounded cases have been checked and used in plots to establish our claims. The presented 2D and 3D plots of solutions of SBS reflect the appearance of a few new two-hump, W-shaped, M-shaped of solution propagation state. We believe these findings are new, not available in the literature for the considered equation.
The results can be helpful for analyzing the dynamics of nonlinear localized waves in the generalized coupled SBS  and other coupled systems.
\section*{Acknowledgements}
The author express his sincere thanks to all the faculty members of department of mathematics, T.D.B. College Raniganj, for their valuable comments.

% BibTeX users please use
% \bibliographystyle{plain}
\bibliographystyle{spmpsci}      % mathematics and physical sciences
\bibliography{SBS}

\begin{thebibliography}{10}
\providecommand{\url}[1]{{#1}}
\providecommand{\urlprefix}{URL }
\expandafter\ifx\csname urlstyle\endcsname\relax
  \providecommand{\doi}[1]{DOI~\discretionary{}{}{}#1}\else
  \providecommand{\doi}{DOI~\discretionary{}{}{}\begingroup
  \urlstyle{rm}\Url}\fi

\bibitem{abbasbandy2009homotopy}
Abbasbandy, S., Magyari, E., Shivanian, E.: The homotopy analysis method for
  multiple solutions of nonlinear boundary value problems.
\newblock Communications in Nonlinear Science and Numerical Simulation
  \textbf{14}(9-10), 3530--3536 (2009)

\bibitem{abdeljawad2015conformable}
Abdeljawad, T.: On conformable fractional calculus.
\newblock Journal of computational and Applied Mathematics \textbf{279}, 57--66
  (2015)

\bibitem{adomian1994solving}
Adomian, G.: Solving frontier problems of physics: the decomposition method,
  with a preface by yves cherruault.
\newblock Fundamental Theories of Physics, Kluwer Academic Publishers Group,
  Dordrecht \textbf{1} (1994)

\bibitem{adomian1983inversion}
Adomian, G., Rach, R.: Inversion of nonlinear stochastic operators.
\newblock Journal of Mathematical Analysis and Applications \textbf{91}(1),
  39--46 (1983)

\bibitem{adomian1993analytic}
Adomian, G., Rach, R.: Analytic solution of nonlinear boundary-value problems
  in several dimensions by decomposition.
\newblock Journal of Mathematical Analysis and Applications \textbf{174}(1),
  118--137 (1993)

\bibitem{adomian1993new}
Adomian, G., Rach, R.: A new algorithm for matching boundary conditions in
  decomposition solutions.
\newblock Applied mathematics and computation \textbf{57}(1), 61--68 (1993)

\bibitem{adomian1994modified}
Adomian, G., Rach, R.: Modified decomposition solution of linear and nonlinear
  boundary-value problems.
\newblock Nonlinear Analysis: Theory, Methods \& Applications \textbf{23}(5),
  615--619 (1994)

\bibitem{arqub2020numerical}
Arqub, O.A., Osman, M.S., Abdel-Aty, A.H., Mohamed, A.B.A., Momani, S.: {A
  Numerical Algorithm for the Solutions of ABC Singular Lane--Emden Type Models
  Arising in Astrophysics Using Reproducing Kernel Discretization Method}.
\newblock Mathematics \textbf{8}(6), 923 (2020)

\bibitem{atangana2013time}
Atangana, A., Secer, A.: {The time-fractional coupled-Korteweg-de-Vries
  equations}.
\newblock In: Abstract and Applied Analysis, vol. 2013. Hindawi (2013)

\bibitem{baldwin2006symbolic}
Baldwin, D., Hereman, W.: {Symbolic software for the Painlev{\'e} test of
  nonlinear ordinary and partial differential equations}.
\newblock Journal of Nonlinear Mathematical Physics \textbf{13}(1), 90--110
  (2006)

\bibitem{baleanu2019investigation}
Baleanu, D., Inc, M., Aliyu, A.I., Yusuf, A.: {The investigation of soliton
  solutions and conservation laws to the coupled generalized
  Schr{\"o}dinger--Boussinesq system}.
\newblock Waves in Random and Complex Media \textbf{29}(1), 77--92 (2019)

\bibitem{biswas2019approximate}
Biswas, S., Ghosh, U., Sarkar, S., Das, S.: {Approximate Solution of Space-Time
  Fractional KdV Equation and Coupled KdV Equations}.
\newblock Journal of the Physical Society of Japan \textbf{89}(1), 014002
  (2020)

\bibitem{chen2018simplest}
Chen, C., Jiang, Y.L.: Simplest equation method for some time-fractional
  partial differential equations with conformable derivative.
\newblock Computers \& Mathematics with Applications \textbf{75}(8), 2978--2988
  (2018)

\bibitem{chowdhury1998painleve}
Chowdhury, A.R., Rao, N.: {Painl{\'e}ve analysis and Backlund transformations
  for coupled generalized Schr{\"o}dinger--Boussinesq system}.
\newblock Chaos, Solitons \& Fractals \textbf{9}(10), 1747--1753 (1998)

\bibitem{das2018rapidly}
Das, P.K.: {Rapidly Convergent Approximation Method to Chiral Nonlinear
  Schrodinger’s Equation in (1+ 2)-dimensions}.
\newblock Sohag J. Math. \textbf{5}, 29--33 (2018)

\bibitem{das2019rapidly}
Das, P.K.: {The rapidly convergent approximation method to solve system of
  equations and its application to the Biswas-Arshed equation}.
\newblock Optik \textbf{195}, 163134 (2019)

\bibitem{das2020chirped}
Das, P.K.: {Chirped and chirp-free optical exact solutions of the Biswas-Arshed
  equation with full nonlinearity by the rapidly convergent approximation
  method}.
\newblock Optik \textbf{223}, 165293 (2020)

\bibitem{das2020new}
Das, P.K.: {New multi-hump exact solitons of a coupled Korteweg-de-Vries system
  with conformable derivative describing shallow water waves via RCAM}.
\newblock Physica Scripta \textbf{95}, 105212 (2020)

\bibitem{das2018piecewise}
Das, P.K., Mandal, S., Panja, M.M.: {Piecewise smooth localized solutions of
  Li{\'e}nard-type equations with application to NLSE}.
\newblock Mathematical Methods in the Applied Sciences \textbf{41}(17),
  7869--7887 (2018)

\bibitem{das2015improved}
Das, P.K., Panja, M.: {An improved Adomian decomposition method for nonlinear
  ODEs}.
\newblock In: Applied Mathematics, pp. 193--201. Springer (2015)

\bibitem{das2016rapidly}
Das, P.K., Panja, M.: A rapidly convergent approximation method for nonlinear
  ordinary differential equations.
\newblock IJSEAS \textbf{2}(8), 334--348 (2016)

\bibitem{das2018solutions}
Das, P.K., Singh, D., Panja, M.: {Solutions and conserved quantities of
  Biswas--Milovic equation by using the rapidly convergent approximation
  method}.
\newblock Optik \textbf{174}, 433--446 (2018)

\bibitem{das2019some}
Das, P.K., Singh, D., Panja, M.M.: {Some modifications on RCAM for getting
  accurate closed-form approximate solutions of Duffing-and Lienard-type
  equations}.
\newblock Journal of Advances in Mathematics \textbf{16}, 8213--8225 (2019)

\bibitem{duan2011new}
Duan, J.S., Rach, R.: {A new modification of the Adomian decomposition method
  for solving boundary value problems for higher order nonlinear differential
  equations}.
\newblock Applied Mathematics and Computation \textbf{218}(8), 4090--4118
  (2011)

\bibitem{ellahi2018exact}
Ellahi, R., Mohyud-Din, S.T., Khan, U., et~al.: {Exact traveling wave solutions
  of fractional order Boussinesq-like equations by applying Exp-function
  method}.
\newblock Results in physics \textbf{8}, 114--120 (2018)

\bibitem{eslami2015soliton}
Eslami, M.: {Soliton-like solutions for the coupled Schrodinger--Boussinesq
  equation}.
\newblock Optik \textbf{126}(23), 3987--3991 (2015)

\bibitem{eslami2016exact}
Eslami, M.: Exact traveling wave solutions to the fractional coupled nonlinear
  schrodinger equations.
\newblock Applied Mathematics and Computation \textbf{285}, 141--148 (2016)

\bibitem{eslami2016first}
Eslami, M., Rezazadeh, H.: The first integral method for wu--zhang system with
  conformable time-fractional derivative.
\newblock Calcolo \textbf{53}(3), 475--485 (2016)

\bibitem{fan2003algebraic}
Fan, E.: An algebraic method for finding a series of exact solutions to
  integrable and nonintegrable nonlinear evolution equations.
\newblock Journal of Physics A: Mathematical and General \textbf{36}(25), 7009
  (2003)

\bibitem{gao2020novel}
Gao, W., Rezazadeh, H., Pinar, Z., Baskonus, H.M., Sarwar, S., Yel, G.: Novel
  explicit solutions for the nonlinear zoomeron equation by using newly
  extended direct algebraic technique.
\newblock Optical and Quantum Electronics \textbf{52}(1), 1--13 (2020)

\bibitem{gepreel2016extended}
Gepreel, K.A.: {Extended trial equation method for nonlinear coupled
  Schrodinger Boussinesq partial differential equations}.
\newblock Journal of the Egyptian Mathematical Society \textbf{24}(3), 381--391
  (2016)

\bibitem{goswami2019efficient}
Goswami, A., Singh, J., Kumar, D., et~al.: An efficient analytical approach for
  fractional equal width equations describing hydro-magnetic waves in cold
  plasma.
\newblock Physica A: Statistical Mechanics and its Applications \textbf{524},
  563--575 (2019)

\bibitem{hon2009series}
Hon, Y., Fan, E.: {A series of exact solutions for coupled Higgs field equation
  and coupled Schr{\"o}dinger--Boussinesq equation}.
\newblock Nonlinear Analysis: Theory, Methods \& Applications \textbf{71}(7-8),
  3501--3508 (2009)

\bibitem{houwe2020solitary}
Houwe, A., Sabi’u, J., Hammouch, Z., Doka, S.Y.: Solitary pulses of a
  conformable nonlinear differential equation governing wave propagation in
  low-pass electrical transmission line.
\newblock Physica Scripta \textbf{95}(4), 045203 (2020)

\bibitem{inc2018dark}
Inc, M., Yusuf, A., Aliyu, A.I., Baleanu, D.: Dark and singular optical
  solitons for the conformable space-time nonlinear schr{\"o}dinger equation
  with kerr and power law nonlinearity.
\newblock Optik \textbf{162}, 65--75 (2018)

\bibitem{jameson2006counting}
Jameson, G.J.: {Counting zeros of generalised polynomials: Descartes’ rule of
  signs and Laguerre’s extensions}.
\newblock The Mathematical Gazette \textbf{90}(518), 223--234 (2006)

\bibitem{khalil2014new}
Khalil, R., Al~Horani, M., Yousef, A., Sababheh, M.: A new definition of
  fractional derivative.
\newblock Journal of Computational and Applied Mathematics \textbf{264}, 65--70
  (2014)

\bibitem{kumar2019hybrid}
Kumar, D., Singh, J., Purohit, S.D., Swroop, R.: A hybrid analytical algorithm
  for nonlinear fractional wave-like equations.
\newblock Mathematical Modelling of Natural Phenomena \textbf{14}(3) (2019)

\bibitem{li2020generalized}
Li, L., Yu, F., Duan, C.: A generalized nonlocal gross--pitaevskii (ngp)
  equation with an arbitrary time-dependent linear potential.
\newblock Applied Mathematics Letters p. 106584 (2020)

\bibitem{liao2020two}
Liao, F., Zhang, L., Wang, T.: {Two energy-conserving and compact finite
  difference schemes for two-dimensional Schr{\"o}dinger-Boussinesq equations}.
\newblock Numerical Algorithms pp. 1--29 (2020)

\bibitem{makhankov1974stationary}
Makhankov, V.: {On stationary solutions of the Schr{\"o}dinger equation with a
  self-consistent potential satisfying Boussinesq's equation}.
\newblock Physics Letters A \textbf{50}(1), 42--44 (1974)

\bibitem{malfliet1992solitary}
Malfliet, W.: Solitary wave solutions of nonlinear wave equations.
\newblock American Journal of Physics \textbf{60}(7), 650--654 (1992)

\bibitem{neirameh2015topological}
Neirameh, A.: {Topological soliton solutions to the coupled
  Schrodinger--Boussinesq equation by the SEM}.
\newblock Optik \textbf{126}(23), 4179--4183 (2015)

\bibitem{olver2000applications}
Olver, P.J.: {Applications of Lie groups to differential equations}, vol. 107.
\newblock Springer Science \& Business Media (2000)

\bibitem{osman2019new}
Osman, M.: {New analytical study of water waves described by coupled fractional
  variant Boussinesq equation in fluid dynamics}.
\newblock Pramana \textbf{93}(2), 26 (2019)

\bibitem{pandir2018analytical}
Pandir, Y., Yildirim, A.: Analytical approach for the fractional differential
  equations by using the extended tanh method.
\newblock Waves in Random and Complex Media \textbf{28}(3), 399--410 (2018)

\bibitem{park2020dynamical}
Park, C., Khater, M.M., Abdel-Aty, A.H., Attia, R.A., Rezazadeh, H., Zidan, A.,
  Mohamed, A.B.: Dynamical analysis of the nonlinear complex fractional
  emerging telecommunication model with higher--order dispersive
  cubic--quintic.
\newblock Alexandria Engineering Journal \textbf{59}(3), 1425--1433 (2020)

\bibitem{prakash2019numerical}
Prakash, A., Goyal, M., Gupta, S.: {Numerical simulation of space-fractional
  Helmholtz equation arising in seismic wave propagation, imaging and
  inversion}.
\newblock Pramana \textbf{93}(2), 28 (2019)

\bibitem{rao1989exact}
Rao, N.: Exact solutions of coupled scalar field equations.
\newblock Journal of Physics A: Mathematical and General \textbf{22}(22), 4813
  (1989)

\bibitem{rao1996coupled}
Rao, N.: Coupled scalar field equations for nonlinear wave modulations in
  dispersive media.
\newblock Pramana \textbf{46}(3), 161 (1996)

\bibitem{rao1997coupled}
Rao, N., Shukla, P.: {Coupled Langmuir and ion-acoustic waves in two-electron
  temperature plasmas}.
\newblock Physics of Plasmas \textbf{4}(3), 636--645 (1997)

\bibitem{rezazadeh2019large}
Rezazadeh, H., Korkmaz, A., Eslami, M., Mirhosseini-Alizamini, S.M.: A large
  family of optical solutions to kundu--eckhaus model by a new auxiliary
  equation method.
\newblock Optical and Quantum Electronics \textbf{51}(3), 1--12 (2019)

\bibitem{sabi2019new}
Sabi’u, J., Jibril, A., Gadu, A.M.: { New exact solution for the (3+ 1)
  conformable space--time fractional modified Korteweg--de-Vries equations via
  Sine-Cosine Method}.
\newblock Journal of Taibah University for Science \textbf{13}(1), 91--95
  (2019)

\bibitem{savaissou2020exact}
Savaissou, N., Gambo, B., Rezazadeh, H., Bekir, A., Doka, S.Y.: Exact optical
  solitons to the perturbed nonlinear schr{\"o}dinger equation with dual-power
  law of nonlinearity.
\newblock Optical and Quantum Electronics \textbf{52}, 1--16 (2020)

\bibitem{wazwaz2001reliable}
Wazwaz, A.: A reliable algorithm for obtaining positive solutions for nonlinear
  boundary value problems.
\newblock Computers \& Mathematics with Applications \textbf{41}(10-11),
  1237--1244 (2001)

\bibitem{wazwaz2000approximate}
Wazwaz, A.M.: Approximate solutions to boundary value problems of higher order
  by the modified decomposition method.
\newblock Computers \& Mathematics with Applications \textbf{40}(6-7), 679--691
  (2000)

\bibitem{wazwaz2000modified}
Wazwaz, A.M.: The modified adomian decomposition method for solving linear and
  nonlinear boundary value problems of tenth-order and twelfth-order.
\newblock International Journal of Nonlinear Sciences and Numerical Simulation
  \textbf{1}(1), 17--24 (2000)

\bibitem{wazwaz2001numerical}
Wazwaz, A.M.: The numerical solution of sixth-order boundary value problems by
  the modified decomposition method.
\newblock Applied Mathematics and Computation \textbf{118}(2-3), 311--325
  (2001)

\bibitem{wazwaz2002numerical}
Wazwaz, A.M.: The numerical solution of special fourth-order boundary value
  problems by the modified decomposition method.
\newblock International journal of computer mathematics \textbf{79}(3),
  345--356 (2002)

\bibitem{xinhui2012homotopy}
Xinhui, S., Liancun, Z., Xinxin, Z., Xinyi, S.: Homotopy analysis method for
  the asymmetric laminar flow and heat transfer of viscous fluid between
  contracting rotating disks.
\newblock Applied Mathematical Modelling \textbf{36}(4), 1806--1820 (2012)

\bibitem{yajima1979soliton}
Yajima, N., Satsuma, J.: Soliton solutions in a diatomic lattice system.
\newblock Progress of Theoretical Physics \textbf{62}(2), 370--378 (1979)

\bibitem{yong2008numerical}
Yong, C., Hong-Li, A.: Numerical solutions of a new type of fractional coupled
  nonlinear equations.
\newblock Communications in Theoretical Physics \textbf{49}(4), 839 (2008)

\bibitem{yu2019inverse}
Yu, F.: Inverse scattering solutions and dynamics for a nonlocal nonlinear
  gross--pitaevskii equation with pt-symmetric external potentials.
\newblock Applied Mathematics Letters \textbf{92}, 108--114 (2019)

\bibitem{yu2020nonstandard}
Yu, F., Fan, R.: Nonstandard bilinearization and interaction phenomenon for
  pt-symmetric coupled nonlocal nonlinear schr{\"o}dinger equations.
\newblock Applied Mathematics Letters \textbf{103}, 106209 (2020)

\end{thebibliography}
\end{document}